\documentclass[a4paper,onecolumn,11pt,unpublished]{quantumarticle}
\pdfoutput=1

\usepackage{amsmath}
\usepackage{amsfonts}
\usepackage{mathtools}
\usepackage{makecell}
\usepackage{multirow}
\usepackage{amsthm}

\usepackage{enumerate}
\usepackage{amssymb}
\usepackage{xparse}
\usepackage{quantikz}
\usepackage[inline,shortlabels]{enumitem}
\usepackage[numbers,sort&compress]{natbib}
\usepackage{hyperref}

\hypersetup{linktocpage,allcolors=blue}

\usepackage[noabbrev,capitalize]{cleveref}

\usepackage{array}
\usepackage{booktabs}
\usepackage{tabularx}

\newcolumntype{L}[1]{>{\raggedright\arraybackslash}p{#1}}
\newcolumntype{Y}{>{\raggedright\arraybackslash}X}

\newtheorem{theorem}{Theorem}
\newtheorem{corollary}{Corollary}[theorem]
\newtheorem{lemma}{Lemma}
\newtheorem{definition}{Definition}
\newtheorem{proposition}{Proposition}
\newtheorem{remark}{Remark}

\newtheorem{problem}{Problem}
\crefname{problem}{problem}{problems}
\Crefname{problem}{Problem}{Problems}

\newcommand{\R}{\mathbb{R}}
\newcommand{\N}{\mathbb{N}}
\newcommand{\C}{\mathbb{C}}

\DeclarePairedDelimiter\abs{\lvert}{\rvert}
\DeclarePairedDelimiter\norm{\lVert}{\rVert}

\DeclarePairedDelimiter\roundBracket{\lparen}{\rparen}
\newcommand{\round}{\roundBracket}
\DeclarePairedDelimiter\multiBracket{[}{]}
\newcommand{\multi}{\multiBracket}

\NewDocumentCommand{\sign}{o}{\IfNoValueTF {#1} {\operatorname{sign}} {\operatorname{sign}\round{#1}} }

\newcommand{\ball}{\mathcal B}
\newcommand{\ballbrack}{\round}

\newcommand{\opn}[1]{\operatorname{#1}}

\DeclarePairedDelimiter{\setbrack}{\{} {\}}
\newcommand{\set}{\setbrack*}

\NewDocumentCommand{\B}{o m m}{\IfNoValueTF {#1} {\ball_{#2} \ballbrack {#3}} {\ball^{#1}_{#2} \ballbrack {#3}}}

\NewDocumentCommand{\F}{o}{\IfNoValueTF {#1} {\operatorname{F}} {\operatorname{F}\multi{#1}} }

\NewDocumentCommand{\reach}{o}{%
    \ensuremath{%
        \IfNoValueTF{#1}{%
            \operatorname{R}%
        }{%
            \operatorname{R}\left[#1\right]%
        }%
    }%
}

\NewDocumentCommand{\chainreach}{o}{%
    \ensuremath{%
        \IfNoValueTF{#1}{%
            \operatorname{CR}%
        }{%
            \operatorname{CR}\left[#1\right]%
        }%
    }%
}

\NewDocumentCommand{\jac}{mmo}{%
    \ensuremath{%
        \IfNoValueTF{#3}{%
            \frac{\partial #1}{\partial #2}%
        }{%
            \left.
            \frac{\partial #1}{\partial #2}
            \right\rvert_{#3}%
        }%
    }%
}

\newcommand{\dd}[1]{\hspace{2pt}d#1}

\usepackage{xcolor}
\usepackage[utf8]{inputenc}

\begin{document}

\title{Quantum query complexity of Hurwitz stability testing}

\author{Maxwell Fitzsimmons}
\orcid{0000-0002-3764-4542}
\email[Corresponding author: ]{mfitzsim@uwaterloo.ca}
\affiliation{University of Waterloo, Waterloo, Ontario N2L 3G1, Canada}

\author{Mohamed Serry}
\email{mserry@uwaterloo.ca}
\affiliation{University of Waterloo, Waterloo, Ontario N2L 3G1, Canada}

\author{Meenu Kumari}
\email{mkumari@uwaterloo.ca}
\affiliation{Digital Technologies, National Research Council Canada, Ottawa, Ontario, Canada}
\affiliation{Institute for Quantum Computing, University of Waterloo, Waterloo, Ontario N2L 3G1, Canada}

\author{Jun Liu}
\email{j.liu@uwaterloo.ca}
\affiliation{University of Waterloo, Waterloo, Ontario N2L 3G1, Canada}

\keywords{quantum algorithms, stability testing, Hurwitz matrices, block-encoding, matrix sign function, quantum query complexity}

\maketitle

\begin{abstract}
Determining whether every trajectory of the linear system \(\dot{x}=Ax\) converges to zero—or, equivalently, whether \(A\) is Hurwitz stable—is a fundamental problem in control theory. 
It arises, for example, in verifying the stability of feedback controllers and state observers, as well as of linearized dynamical models.
In this work, we study the quantum query complexity of deciding whether a matrix \(A\in\mathbb{C}^{n\times n}\) is Hurwitz stable, given an exact normalized block-encoding of \(A\). 
The query complexity is characterized in terms of the system dimension \(n\) and a supplied lower bound \(0<\nu\leq \omega_A\) on the imaginary-axis distance
\(
\omega_A:=
\min_{t\in\mathbb{R}}
\sigma_{\min}(A-\mathrm{i}tI)\), where $\sigma_{\min}(A)$ is the smallest singular value of $A$.
We prove a query-complexity lower bound of  \(\Omega(\sqrt{n}/\nu)\), which holds even when \(A\) is restricted to be real and diagonal. 
For general matrices with nonzero $\omega_A$, we encode an approximation of the right-half-plane spectral projector \(\sign^+(A):=(I+\operatorname{sign}(A))/2\) as a block of a matrix inverse; $A$ is Hurwitz if and only if this projector is zero. 
We give a quantum algorithm that, from a block-encoding of an idempotent matrix with normalization $\beta$ and error at most $1/(4\sqrt{n})$, decides whether the matrix is zero, and we show that its $\Theta(\beta\sqrt{n})$ query complexity is optimal for constant success probability. 

When the zero test is applied to specific approximations of $\sign^+(A)$, this yields an upper bound of \(O(\sqrt{n}\,\nu^{-2}\log(\sqrt{n}/\nu))\) queries in the general case, and \(O(\sqrt{n}\,\nu^{-1}\log n)\) when $A$ is Hermitian or normal. 
The latter upper bound matches the lower bound up to a logarithmic factor. 
Thus the dimension dependence of quantum Hurwitz testing is nearly optimal; whether the additional factor of $1/\nu$ in the non-normal case is necessary remains open.
\end{abstract}
\section{Introduction}\label{sec:into}

Do all trajectories of a dynamical system converge to zero from every initial condition? For the continuous-time linear system
\begin{equation}
    \dot{x}=Ax,
    \label{eq:intro-linear-system}
\end{equation}
the answer is affirmative exactly when every eigenvalue of \(A\) lies
in the open left half-plane. In this case,  \(A\) is said to be \textit{Hurwitz stable}, or simply Hurwitz.
Verifying Hurwitz stability is a fundamental task that appears consistently in systems and control. 
A  feedback gain matrix \(K\)  stabilizes the linear control system $\dot{x}=\bar{A}x+Bu$ under the state-feedback law $u=Kx$  precisely
when the closed-loop matrix \(\bar{A}+BK\) is Hurwitz. Also, for the system \(\dot{x}=\bar{A}x\) with measurement
\(y=Cx\), a proposed observer gain matrix \(L\) in
\(\dot{\tilde{x}}=\bar{A}\tilde{x}+L(y-C\tilde{x})\) ensures that the
estimation error \(x-\tilde{x}\) converges to zero precisely when
\(\bar{A}-LC\) is Hurwitz \cite{AstromMurray2021}. Furthermore, verifying Hurwitz stability is central to the local analysis of
nonlinear systems. For a nonlinear system \(\dot{x}=f(x)\) with an
equilibrium \(f(x^\star)=0\), the eigenvalues of the Jacobian
\(Df(x^\star)\) determine the local asymptotic stability of \(x^\star\)
whenever the equilibrium is hyperbolic \cite{AstromMurray2021}.
In this work, we focus on the decision problem of verifying asymptotic stability: whether \(A\), or \(Df(x^\star)\) at an equilibrium \(x^\star\), is Hurwitz.

  Quantum algorithms for analyzing dynamical systems, such as those modeled by
differential equations, typically encode a solution over a finite time
horizon. Stability can improve the complexity of such simulations and may
even permit time fast-forwarding \cite{JenningsEtAl2024}. Hurwitz testing
poses a different question: it produces a qualitative decision that applies
to every initial condition, without requiring the preparation or simulation
of a particular trajectory. We ask:
\begin{quote}
\emph{How many queries to a block-encoding of matrix \(A\) are necessary and
sufficient to decide whether \(A\) is Hurwitz?}
\end{quote}
We answer this question given a known lower bound, or a promise, $\nu>0$ on the imaginary-axis
distance
\[\omega_A :=
    \min_{t\in\mathbb{R}}
    \sigma_{\min}(A-\mathrm{i}tI).
\]
Our results apply to arbitrary matrices \(A\), including nonnormal and
nondiagonalizable matrices.

\subsection{Main results, algorithmic strategy, and scope}

To determine asymptotic stability without explicitly inspecting the
eigenvalues of \(A\), our algorithms use the spectral projector onto the
right-half-plane invariant subspace.  Define the scalar sign by
\[
    \opn{sign}(z)
    :=
    \begin{cases}
        1, & \Re(z)>0,\\
        -1, & \Re(z)<0,
    \end{cases}
\]
and consider a Jordan decomposition
\[
    A=V\operatorname{diag}\bigl(
        J_{r_1}(\lambda_1),\ldots,J_{r_q}(\lambda_q)
    \bigr)V^{-1}.
\]
Then, the matrix sign of $A$ is defined as:
\[
    \opn{sign}(A)
    :=
    V\operatorname{diag}\bigl(
        \opn{sign}(\lambda_1)I_{r_1},\ldots,
        \opn{sign}(\lambda_q)I_{r_q}
    \bigr)V^{-1}.
\]
The matrix
\begin{equation}
  \label{eq:intro-projector}
    \sign^+(A):=\frac{I+\opn{sign}(A)}{2} 
\end{equation}
projects onto the direct sum of the generalized eigenspaces associated
with eigenvalues in the open right half-plane. Therefore, 
\[
    A\text{ is Hurwitz}
    \iff
    \sign^+(A)=0.
\]
Thus, the decision problem reduces to testing whether a block-encoded
idempotent is zero.

 By Proposition~\ref{thm:omega_A}, we may assume without loss of generality
that \(A\) is encoded exactly; that is, the input oracle \(U_A\) is a
\((1,a,0)\)-block-encoding of \(A\). The main results on the query complexity of deciding Hurwitz stability are
summarized in Table~\ref{tab:hurwitz-query-complexity}, in terms of the system dimension $n$ and the lower bound $\nu$ on the imaginary-axis distance $\omega_{A}$. In particular,
Theorem~\ref{thm:diagonal-Hurwitz-complexity} establishes the lower bound
\[
\Omega(\sqrt{n}/\nu),
\]
which holds even when \(A\) is restricted to be real and diagonal. For
constant success probability and general \(A\), we prove the upper bound
(\cref{thm:determine_stable_encoding_ours})
\[
O(\sqrt{n}\,\nu^{-2}\log(\sqrt{n}/\nu)),
\]
whereas, for Hermitian or normal \(A\), we prove the upper bound
(\cref{thm:determine_stable_encoding_normal})
\[
O(\sqrt{n}\,\nu^{-1}\log n).
\]
The latter matches the lower bound up to a logarithmic factor. Thus, the dimension dependence of quantum Hurwitz-stability testing is nearly
optimal. For general nonnormal matrices, an additional factor of
\(\nu^{-1}\) remains between the upper and lower bounds; whether this
dependence on \(1/\nu\) is necessary is an open question.

	\begin{table*}[t]
		\centering
		\small
		\renewcommand{\arraystretch}{1.25}
		\begin{tabular}{@{}llll@{}}
			\toprule
			\parbox[t]{0.22\textwidth}{\raggedright Input oracle and setting}
			& \parbox[t]{0.25\textwidth}{\raggedright Upper bound}
			& \parbox[t]{0.18\textwidth}{\raggedright Lower bound}
			& \parbox[t]{0.25\textwidth}{\raggedright Comparison} \\
			\midrule

			\parbox[t]{0.22\textwidth}{\raggedright Supplied \((\beta,b,\epsilon_+)\)-block-encoding \(U_P\) of
			\(\sign^+(A)\), with
			\(b\geq1\) and \(\beta+\epsilon_+\geq1\)
			\par\smallskip
			Oracle: \(U_P\)}
			&
			\parbox[t]{0.25\textwidth}{\raggedright \(\displaystyle O(\beta\sqrt n)\)
			\par\smallskip
			\cref{thm:determine_stable_encoding_supplied}}
			&
			\parbox[t]{0.18\textwidth}{\raggedright \(\displaystyle \Omega(\beta\sqrt n)\)
			\par\smallskip
			\cref{thm:idempotent-zero-lower-bound}}
			&
			\parbox[t]{0.25\textwidth}{\raggedright Optimal dependence on \(n\) and \(\beta\).} \\

			\addlinespace

			\parbox[t]{0.22\textwidth}{\raggedright Normal \(A\) (including Hermitian \(A\))
			\par\smallskip
			Oracle: \(U_A\)}
			&
			\parbox[t]{0.25\textwidth}{\raggedright \(\displaystyle
			O\left(\frac{\sqrt n}{\nu}\log n\right)
			\)
			\par\smallskip
			\cref{thm:determine_stable_encoding_normal}}
			&
			\parbox[t]{0.18\textwidth}{\raggedright \(\displaystyle
			\Omega\left(
			\frac{\sqrt n}{\nu}
			\right)
			\)
			\par\smallskip
			\cref{thm:diagonal-Hurwitz-complexity}}
			&
			\parbox[t]{0.25\textwidth}{\raggedright Bounds match up to a factor \(\log n\).} \\

			\addlinespace

			\parbox[t]{0.22\textwidth}{\raggedright General \(A\) (possibly non-normal)
			\par\smallskip
			Oracle: \(U_A\)}
			&
			\parbox[t]{0.25\textwidth}{\raggedright \(\displaystyle
			O\left(
			\frac{\sqrt n}{\nu^2}
			\log\frac{\sqrt n}{\nu}
			\right)
			\)
			\par\smallskip
			\cref{thm:determine_stable_encoding_ours}}
			&
			\parbox[t]{0.18\textwidth}{\raggedright \(\displaystyle
			\Omega\left(
			\frac{\sqrt n}{\nu}
			\right)
			\)
			\par\smallskip
			\cref{thm:diagonal-Hurwitz-complexity}}
			&
			\parbox[t]{0.25\textwidth}{\raggedright A factor \(1/\nu\) remains between the bounds, apart
			from logarithms.} \\

			\bottomrule
		\end{tabular}
		\par\smallskip
		\begin{minipage}{\textwidth}
			\footnotesize
			\textit{Conventions.} Here \(n\geq2\), \(A\) has no purely
			imaginary eigenvalues, and the failure probability is at most \(1/3\).
			Queries are to the indicated oracle or its adjoint;
			controlled calls also count as one query.
			\par\smallskip
			\textit{Supplied projector.} In the first row,
			\(0\leq\epsilon_+\leq1/(4\sqrt n)\).
			If \(b=0\) or \(\beta+\epsilon_+<1\), the optimal query count
			is \(0\).
			\par\smallskip
			\textit{Input matrix.} In the last two rows, \(U_A\) is an
			exact \((1,a,0)\)-block-encoding of \(A\), \(\norm{A}\leq1\),
			and a known \(0<\nu\leq\omega_A\) is supplied after the input reduction.
			These lower bounds already hold for real diagonal \(A\)
			with \(a=1\).
		\end{minipage}

		\caption{Worst-case oracle-query bounds for Hurwitz decision
			with success probability at least \(2/3\).}
		\label{tab:hurwitz-query-complexity}
	\end{table*}

  The upper bounds are obtained in two steps. First, assuming access to a \((\beta,b,\epsilon_{+})\)-block encoding of the
projector \(\sign^{+}(A)\), we establish the sharp query-complexity bound
\[
\Theta(\beta\sqrt{n});
\]
see \cref{thm:determine_stable_encoding_supplied}. The algorithm uses
amplitude estimation to decide whether this projector is the zero matrix. For non-normal matrices, the projector
block encoding is constructed using a Pad\'e approximation implemented
through a structured matrix inverse; see
\cref{sec:matrix_sign}. The advantage of Pad\'e-based
construction over alternative implementations of matrix functions is outside the scope of this work. 

The remainder of this paper is organized as follows. 
\Cref{sec:related_work} discusses related work.
\Cref{sec:problem} introduces the decision problem and access model.
\Cref{sec:worst_complexity} establishes the query-complexity lower bounds.
\Cref{sec:SignEncodingSupplied} analyzes the determination of stability given a block encoding of the spectral projector. 
\cref{sec:determining_stability} provides an end-to-end query-complexity analysis for our algorithm which decides if a given matrix is Hurwitz. 
\Cref{sec:Discussion} discusses the implications and limitations of the results, and \cref{sec:conclusion} concludes the paper.

\subsection{Related work}\label{sec:related_work}

For a Hurwitz matrix, $\omega_A$ is its complex, unstructured stability radius \cite{HinrichsenPritchard1986}.
It captures perturbation sensitivity that an eigenvalue gap alone does not describe for non-normal matrices \cite{TrefethenEtAl1993}.
The comparisons below concern the guarantees obtainable from normalized block-encoding access and a supplied margin $\nu\leq\omega_A$.

\subsubsection{Quantum stability testing and eigenvalue estimation}

Kolada and Stefa\'nski test stability through roots of generalized characteristic equations, including those arising from fractional-order systems \cite{KoladaStefanski2025Stability}.
They map the right half-plane to the unit disc and use phase samples on a triangular mesh to identify candidate edges, combining Grover search, quantum counting and local refinement.
Their $O(\sqrt P)$ bound concerns the initial candidate-edge search over $P$ mesh edges \cite[Section~II.B]{KoladaStefanski2025Stability}.
The required mesh size and the cost of implementing its oracle from a matrix block-encoding are not bounded in terms of $n$ and $\nu$; consequently, that search estimate does not give a matrix-query bound under our promise.

For block-encoded matrices, Zhang \emph{et al.} estimate an eigenvalue closest to a prescribed line \cite[Problem~1]{ZhangEtAl2025NonHermitian}.
This also gives a route to stability testing: when $\norm{A}\leq1$, the distance of an eigenvalue $\lambda$ to the line $\operatorname{Re}z=2$ is $2-\operatorname{Re}\lambda$, so a closest eigenvalue has largest real part.
Lin and Shang use ground-energy estimation for Hermitian matrices constructed from $A-\mu I$ and explicitly propose searching the right half-plane to test stability \cite[Section~5]{LinShang2026NonNormal}.

Both eigensolver analyses use eigenvector or Jordan-basis conditioning to infer eigenvalue locations from small singular values of shifted matrices.
Their costs also depend on prepared-state overlap and, for defective matrices, Jordan-block sizes; Lin and Shang also state a spectral-gap assumption for the shifted Hermitian problems, which appears only in their eigenvector-preparation bound \cite[Section~4 and Theorem~3]{LinShang2026NonNormal}.
Generic initialization can supply overlap at a dimension-dependent cost, but the conditioning dependence remains.
Even the best-conditioned eigenvector basis can become arbitrarily ill-conditioned at fixed dimension while $\omega_A$ stays bounded away from zero (\cref{sec:omega_eigenvector_conditioning}).
Thus their stated estimation bounds do not directly yield a uniform bound in $n$ and $\nu$ alone.
Our decision bound in \cref{tab:hurwitz-query-complexity} remains uniform on these families, distinguishing the guarantees under a common matrix-access model.

Zhang, Wu and Yuan instead extract eigenvalues from a signal with a small number of spectral components.
Their stability application requires an eigenvalue with largest real part to contribute detectably to that signal \cite[Appendix~E.2]{ZhangWuYuan2025Eigenvalue}.
Their estimates therefore rest on sparsity, visibility and separation promises about a prepared state; our decision bound requires no such state.

\subsubsection{Quantum query bounds and spectral counting}

Lin and Shang identify query lower bounds as future work \cite[Section~7]{LinShang2026NonNormal}.
Zhang \emph{et al.} prove BQP-completeness for a guided decision problem distinguishing nearly real from detectably complex spectra \cite[Problem~2]{ZhangEtAl2025NonHermitian}.
Zhang, Wu and Yuan discuss precision optimality through Heisenberg-limited scaling \cite[Appendix~D.4.1]{ZhangWuYuan2025Eigenvalue}.
These results do not provide quantitative quantum query lower bounds for stability under our margin promise.
Our lower bounds in \cref{tab:hurwitz-query-complexity} complement these algorithmic and hardness results by quantifying the joint dimension and margin dependence of the stability decision.
Because every eigenvalue satisfies $|\Re\lambda|\geq\omega_A\geq\nu$ under the promise, estimating the spectral abscissa to accuracy below $\nu$ decides stability, so the $\Omega(\sqrt n/\nu)$ bound applies equally to that estimation task when no guiding state is supplied.

For normal $A$, the eigenvalues of $H=(A+A^\dagger)/2$ are the real parts of those of $A$, reducing Hurwitz testing to an extremal Hermitian eigenvalue decision.
Kerzner \emph{et al.} obtain $\widetilde O(\sqrt n/\epsilon)$ queries for extremal-eigenvalue estimation to additive accuracy $\epsilon$ with normalized sparse-matrix access \cite{KerznerEtAl2024SmallestEigenvalue}.
Mande and de Wolf establish upper and lower bounds with the same dimension and inverse-precision dependence, up to logarithmic factors, for extremal eigenphase estimation with unitary access \cite{MandeDeWolf2026PhaseEstimation}.
These bounds require no supplied guiding state and establish the scaling recovered by our normal-matrix bounds.
Mande and de Wolf's lower bounds use fractional search.
Our matrix and idempotent lower bounds instantiate the same obstruction for block-encoding oracles using the hybrid method of \cite{BennettEtAl1997}.
Jeffery and Witteveen remove an overlap-dependent logarithmic overhead for eigenphase-threshold testing \cite[Corollary~4.1]{JefferyWitteveen2026GroundState}; its possible use here is discussed in \cref{rem:normal_eigenphase_extension}.

Futamura, Ye and Sakurai use contour approximations of spectral projectors and measured norms to estimate Hermitian eigenvalue counts \cite{FutamuraYeSakurai2021EigenvalueDensity}.
Zero testing asks only whether the spectral subspace is empty, but our idempotent formulation also permits oblique projectors and approximate block-encodings.
Under the stated approximation promise, \cref{thm:zero_test_idempotent,thm:idempotent-zero-lower-bound} give matching dimension and normalization dependence at fixed failure probability.

\subsubsection{Matrix-function constructions}

Matrix-sign methods for separating invariant subspaces are classical \cite{ROBERTS01101980,higgs09}.
Quantum constructions include  polynomial eigenvalue transformations \cite{LowSu2024QEVT}, block-encodings of rational sums of shifted inverses \cite{TakahiraOhashiSogabeUsuda2022ContourBlockEncoding}, matrix-sign embeddings \cite{wang2026signembeddingquantumalgorithms}, and Riesz or oblique spectral projectors \cite{rodenasruiz2026quantumsolversnonlinearmatrix,DalzellSu2026ObliqueEigenprojection}.
Our upper-bound proof combines a spectral-projector approximation with amplitude estimation \cite{BrassardGilles2000QAAa}, tracking normalization, approximation error and construction cost under $\nu\leq\omega_A$.
We use the direct Pad\'e realization in \cref{sec:matrix_sign}, whose normalization and query cost are stated in terms of $\nu$, to make this accounting explicit, including circuit costs; we do not claim that it improves on the cited alternatives.

\section{Problem formulation and access model}\label{sec:problem}

We write $\N$ for the positive integers. For vectors, $\norm{\cdot}$
is the Euclidean norm; for matrices, it is the induced operator norm.
For $A\in\C^{n\times n}$, write $A^\dagger$ for its conjugate transpose,
$\opn{spec}(A)$ for its spectrum, $\sigma_{\min}(A)$ and
$\sigma_{\max}(A)$ for its extremal singular values, and
$\norm{A}_F=\sqrt{\opn{tr}(A^\dagger A)}$ for its Frobenius norm.
Quantum states are unit vectors in finite-dimensional complex
inner-product spaces. We use bra--ket notation and tensor products for
joint registers; identities on unaffected registers are implicit.

\begin{definition}[Block-encoding]\label{def:block-encoding}
Let $A\in\C^{n\times n}$, $\alpha>0$, $\epsilon\geq0$, and
$a\in\N\cup\{0\}$. A unitary $U\in\C^{n2^a\times n2^a}$ is an
$(\alpha,a,\epsilon)$-block-encoding of $A$ if
\[
 \norm{A-\alpha(\bra{0^a}\otimes I_n)U(\ket{0^a}\otimes I_n)}
 \leq\epsilon.
\]
The encoding is exact when $\epsilon=0$. 
If $a=0$ we adopt the convention that $\bra{0^a}\otimes I_n=I_n$. 
\end{definition}

This is the block-encoding convention of \cite{QSVT}; in particular,
$\epsilon$ measures error in $A$, rather than in $A/\alpha$.
The input model supplies the indicated block-encoding oracle, its
adjoint, and their controlled versions. Each application counts as
one query. Controlled access is an assumption on the supplied oracle.
The dimension, encoding parameters, and other classical parameters
specified in a result are known. The initial state and non-oracle
operations may depend on these parameters, but not on the unknown
matrix. Ancillary registers and intermediate measurements are allowed.
Query complexity is the maximum number of oracle calls over admissible
inputs and measurement outcomes. Costs of implementing the input
oracle are separate from this count. We interpret logarithmic factors
in asymptotic bounds as being at least one.

A matrix is Hurwitz if every eigenvalue has negative real part.
Equivalently, every solution of $\dot x=Ax$ converges to zero.   
We first introduce the general Hurwitz decision problem and then motivate a
modified version in which a lower bound $\nu$ on the imaginary-axis distance
\(\omega_A\) is provided.

\begin{problem}[Hurwitz decision problem]\label{prob:Hurwitz_decision_hard}
Given $A\in\C^{n\times n}$ with no eigenvalues on the imaginary axis,
determine whether $A$ is Hurwitz.
\end{problem}
We remark that the assumption of no purely imaginary eigenvalues cannot be dropped when the entries of $A$ are real numbers given through approximations: already for $n=1$, deciding whether $a\in\R$ satisfies $a<0$ or $a\geq0$ is not computable, because the sign is a discontinuous function of $a$ and computable functions on the reals are continuous \cite{Weihrauch2000}.

For $A\in\C^{n\times n}$, recall the definition of the imaginary-axis distance
\[
 \omega_A=\min_{t\in\R}\sigma_{\min}(A-\mathrm{i}tI).
\]
The minimum exists because the function being minimized is continuous and is at least $|t|-\norm A$. 
The distance to a singular matrix gives
\[
 \omega_A
 =\inf\bigl\{\norm E:
     \opn{spec}(A+E)\cap\mathrm{i}\R\ne\varnothing\bigr\}.
\]
Thus $\omega_A>0$ exactly when $A$ has no purely imaginary eigenvalues.
When $A$ is Hurwitz, $\omega_A$ is its complex, unstructured stability
radius \cite{HinrichsenPritchard1986}; for a general $A$, it measures
distance to an imaginary-axis eigenvalue, not necessarily distance to
the boundary of the Hurwitz-stable set. Moreover,
\[
 \omega_A\leq\sigma_{\min}(A),\qquad
 \omega_A\leq\min_{\lambda\in\opn{spec}(A)}|\Re\lambda|.
\]
If $A$ is normal, meaning $AA^\dagger=A^\dagger A$, then the second inequality is an equality, $\omega_A=\min_{\lambda\in\opn{spec}(A)}|\Re\lambda|$. 
For non-normal matrices, the inequalities can be strict.

The matrices $A_n=-I_n/n$ illustrate why a quantitative promise $\nu$ is
needed for uniform query bounds: they are Hurwitz, but
$\omega_{A_n}=1/n\to0$. We study the following problem with a
supplied lower bound on $\omega_{A}$.
\begin{problem}[Hurwitz decision problem with known gap]\label{prob:Hurwitz_decision}
Given $A\in\C^{n\times n}$ and a known $\nu>0$ satisfying
$\nu\leq\omega_A$, determine whether $A$ is Hurwitz.
\end{problem}

The following proposition allows us to consider $(1,a,0)$-block-encodings of the input matrix $A$ without losing generality for \cref{prob:Hurwitz_decision}.

	\begin{proposition}[Robustness of input block-encodings; see the proof in \cref{sec:Proof_thm:omega_A}]
		\label{thm:omega_A}
		Let \(A\in\C^{n\times n}\) have no purely imaginary eigenvalues.
		Let \(U\) be an \((\alpha,a,\epsilon)\)-block-encoding of \(A\),
		where \(\epsilon\geq0\), and define
		\[
		\widetilde A_U
		=(\bra{0^a}\otimes I_n)U(\ket{0^a}\otimes I_n).
		\]
		Then \(U\) is an exact \((1,a,0)\)-block-encoding of
		\(\widetilde A_U\), and the following hold.
		\begin{enumerate}
			\item
			\[
			\omega_{\widetilde A_U}
			\geq \frac{\omega_A-\epsilon}{\alpha}.
			\]
			
			\item
			If \(\epsilon<\omega_A\), then \(\widetilde A_U\) has no
			purely imaginary eigenvalues, and
			\(\sign(A)\) and \(\sign(\widetilde A_U)\) are similar. Thus, they have the same number of eigenvalues with positive and negative real part. 
			In particular,
			\[
			A\text{ is Hurwitz}
			\quad\Longleftrightarrow\quad
			\widetilde A_U\text{ is Hurwitz}.
			\]

			\item
			Suppose \(a\geq1\).
			For every \(\omega>0\) with \(\epsilon\geq\omega\), there
			exist matrices \(B,C\in\C^{n\times n}\) and a unitary
			\(V\) such that \(B\) is Hurwitz, \(C\) has an eigenvalue
			with positive real part, \(\omega_B=\omega_C=\omega\),
			and \(V\) is an \((\alpha,a,\epsilon)\)-block-encoding
			of both \(B\) and \(C\).
			Thus, when \(\epsilon\geq\omega\), a block-encoding and
			its parameters need not determine whether the encoded
			matrix is Hurwitz, even when the margin $\omega$ is known.
		\end{enumerate}
	\end{proposition}
	Item 3 shows that an algorithm that receives only a block-encoding and its parameters can solve \cref{prob:Hurwitz_decision} only if the encoding error is below the margin. Item 2 shows that, within this allowable error, deciding whether $A$ is Hurwitz is the same as deciding whether $\widetilde A_U$ is Hurwitz. 
	Henceforth, we will assume that $U_A$ is a $(1,a,0)$-block-encoding of $A$.

\section{Lower complexity bounds}\label{sec:worst_complexity}
		The following lemma, whose proof is given in \cref{sec:Proof_thm:search-hybrid}, provides a basis for the complexity lower bounds in this paper.  
		
	\begin{lemma}[Lower bound for search with weak unitary oracles]
		\label{thm:search-hybrid}
		Let $n,d\in\N$, let $V_0,V_1\in\C^{d\times d}$ be unitary,
		and set
		\[
		\eta=\norm{V_1-V_0}>0.
		\]
		Write $\mathbf 0$ for the all-zero bit string in $\{0,1\}^n$
		and $e_j$ for the bit string whose only nonzero entry is
		its $j$-th entry. For $x\in\{\mathbf 0,e_1,\ldots,e_n\}$,
		define
		\[
		U_x=\sum_{k=1}^n V_{x_k}\otimes\ket{k}\bra{k},
		\]
		 where \(x_k\) denotes the \(k\)-th bit of the string \(x\), so that
\(V_{x_k}\in\{V_0,V_1\}\), and  $\ket{1},\ldots,\ket{n}$ are the standard basis
		vectors of the index register $\C^n$.
		
		For every fixed $\delta\in(0,1/2)$, any quantum algorithm
		that, given query access to $U_x$ and $U_x^\dagger$,
		including controlled applications, distinguishes
		\[
		x=\mathbf 0
		\qquad\text{from}\qquad
		x\in\{e_1,\ldots,e_n\}
		\]
		with probability at least $1-\delta$ on every allowed
		input requires
		\[
		\Omega\left(\frac{\sqrt n}{\eta}\right)
		\]
		queries in the worst case. The implied constant depends
		only on $\delta$. The algorithm may depend on the parameters defining the oracle
		family and on any other classical data that are the same for
		every $x$ in that family.
	\end{lemma}

	We argue now that the $\sqrt{n}$ complexity is necessary for \cref{prob:Hurwitz_decision}.

	\begin{theorem}
		\label{thm:diagonal-Hurwitz-complexity}
		Let
		\[
		\mathcal D_{n,\nu}
		=
		\left\{
		A=\operatorname{diag}(a_1,\ldots,a_n)\in\mathbb R^{n\times n}:
		\nu\leq\abs{a_j}\leq1
		\right\},
		\qquad 0<\nu\leq1.
		\]
		
		Any quantum algorithm which, given query access to an exact
		\((1,a,0)\)-block-encoding of an unknown
		\(A\in\mathcal D_{n,\nu}\), decides, with success probability at least $1-\delta$ for a fixed $\delta\in(0,1/2)$, 
		whether \(A\) is Hurwitz has worst-case query
		complexity
		\[
		\Omega\left(\frac{\sqrt n}{\nu}\right).
		\]
		This lower bound already holds for \(a=1\).
	\end{theorem}
	
	\begin{proof}
		For \(x\in\set{0,e_1,\ldots,e_n}\), define
		\[
		A_x
		=
		\nu\operatorname{diag}(2x_1-1,\ldots,2x_n-1).
		\]
		Then \(\norm{A_x}=\omega_{A_x}=\nu\), while \(A_0\) is Hurwitz and
		each \(A_{e_j}\) is not Hurwitz. For \(t\in[-1,1]\), let
		\[
		R(t)=
		\begin{pmatrix}
			t&\sqrt{1-t^2}\\
			\sqrt{1-t^2}&-t
		\end{pmatrix}.
		\]
		 	Take
		$V_0=R(-\nu)$ and
		$V_1=R(\nu)$
		in \cref{thm:search-hybrid}, where
		$\eta
		=
		\norm{V_{1}-V_{0}}
		=
		2\nu$.  The corresponding oracle is
		$U_x
		=
		\sum_{k=1}^n
		V_{x_{k}}\otimes\ket{k}\bra{k}$,
		which is an exact \((1,1,0)\)-block-encoding of \(A_x\).  By \cref{thm:search-hybrid}, for every fixed \(\delta\in(0,1/2)\), the worst-case
		query complexity is
		$\Omega\round*{\frac{\sqrt n}{\nu}}$.
		
	\end{proof}

\section{Determining stability using the block-encoding of $\sign^{+}(A)$}
\label{sec:SignEncodingSupplied}
	The right-half-plane spectral projector $\sign^+(A)$ is idempotent, and $A$ is Hurwitz if and only if $\sign^+(A)=0$. A matrix $P$ is idempotent if $P^2=P$; a nonzero idempotent has an eigenvalue equal to $1$, so $P=0$ or $\norm{P}\geq1$. Hurwitz testing from a block-encoding of $\sign^+(A)$ therefore reduces to deciding whether a block-encoded idempotent is zero. In this section, we analyze the query complexity of determining whether an idempotent matrix is zero, given a block encoding of that matrix. We then extend this result to decide Hurwitz stability from a block encoding of the matrix sign function.
	
	The supplied block-encoding parameters can themselves determine
	whether an idempotent is zero. 
	The following result characterizes these edge cases and gives a lower bound outside of them.
	
	\begin{theorem}[Zero-testing block-encoded idempotents]
		\label{thm:idempotent-zero-lower-bound}
		Let $n\in\N$, $b\in\N\cup\{0\}$, $\beta>0$, $\epsilon\in [0,0.5)$
		and $\delta\in(0,1/2)$.
		Suppose that an unknown idempotent $P\in\C^{n\times n}$
		is supplied through a $(\beta,b,\epsilon)$-block-encoding
		$U_P$. An idempotent in $\C^{n\times n}$ is called
		admissible if it admits a block-encoding with these same
		parameters. The following hold.
		\begin{enumerate}
			\item \label{item:idempotent-admissibility}
			Both $P=0$ and $P\neq0$ occur among admissible
			idempotents if and only if
			\[
			b\geq1
			\qquad\text{and}\qquad
			\beta+\epsilon\geq1.
			\]

			\item \label{item:idempotent-query-complexity}
			Let $T$ be the optimal worst-case number of queries
			to $U_P$ and $U_P^\dagger$ required to decide whether
			$P=0$, with success probability at least $1-\delta$
			for every admissible $P$ and every such encoding $U_P$.
			The parameters
			$n,\beta,b,\epsilon,\delta$ are given to the algorithm,
			and controlled queries are allowed.
			For fixed $\delta$,
			\[
			T=
			\begin{cases}
				\Omega(\beta\sqrt n),
				& b\geq1 \text{ and } \beta+\epsilon\geq1,\\[1mm]
				0,
				& \text{otherwise}.
			\end{cases}
			\]
			The implied constant depends only on $\delta$.
			The lower bound already holds for $b=1$ and for
			diagonal orthogonal projectors of rank zero or one.
			
			In the zero-query cases, every admissible $P$ equals
			$0$ if $\beta+\epsilon<1$, and equals $I_n$ if
			$b=0$ and $\beta+\epsilon\geq1$.
		\end{enumerate}
	\end{theorem}
	
	\begin{proof}
		Write
		\[
		B=(\langle {0^b} \rvert\otimes I_n)U_P
		(\lvert {0^b} \rangle\otimes I_n).
		\]
		Since $B$ is a compression of a unitary,
		$\norm{B}\leq1$. If $P\neq0$, idempotence implies that
		$P$ has an eigenvalue equal to $1$, so
		\[
		1\leq\norm{P}
		\leq\norm{P-\beta B}+\beta\norm{B}
		\leq\epsilon+\beta.
		\]
		Hence $\beta+\epsilon<1$ forces $P=0$.
		
		Suppose next that $b=0$ and $\beta+\epsilon\geq1$.
		Then $B=U_P$ and since $\epsilon<0.5$, 
		$\beta\geq1-\epsilon>\epsilon$.
		If a unit vector $\ket{v}$ satisfies $P\ket{v}=0$, then
		\[
		\beta
		=\norm{\beta U_P\ket{v}}
		=\norm{(P-\beta U_P)\ket{v}}
		\leq\epsilon,
		\]
		a contradiction. Thus $P$ is invertible, and $P^2=P$
		implies $P=I_n$. 
		This completes the $\implies$ case for item 1, and proves the zero-query assertion in item 2.

		Now fix $b,\beta$ with $b\geq1$ and $\beta+\epsilon\geq1$.
		Set
		\[
		s=\min\{1,\beta^{-1}\}.
		\]
		Then $0<s\leq1$. The quantity $\abs{1-\beta s}$ is $0$ for $\beta\geq1$, and equals
		$1-\beta\leq\epsilon$ otherwise. Hence, $\abs{1-\beta s}\leq\epsilon$. 
		
		For $0\leq t\leq1$, define the unitary
		\[
		R(t)=
		\begin{pmatrix}
			t & \sqrt{1-t^2}\\
			\sqrt{1-t^2} & -t
		\end{pmatrix},
		\]
		and set
		\[
		V_\ell=I_{2^{b-1}}\otimes R(s\ell),
		\qquad \ell\in\{0,1\}.
		\]
		As in \cref{thm:search-hybrid}, let $\mathbf0$ be the
		all-zero bit string and $e_j$ the bit string whose only
		nonzero entry is its $j$-th entry. For
		$x\in\{\mathbf0,e_1,\ldots,e_n\}$, define
		\[
		P_x=\sum_{k=1}^n x_k\ket{k}\bra{k},
		\qquad
		U_x=\sum_{k=1}^n V_{x_k}\otimes\ket{k}\bra{k}.
		\]
		Each $P_x$ is a diagonal orthogonal projector of rank
		zero or one, so $\norm{P_x}\leq1$. Each $U_x$ is unitary,
		and
		\[
		(\langle {0^b} \rvert\otimes I_n)U_x
		(\lvert {0^b} \rangle\otimes I_n)=sP_x.
		\]
		Consequently,
		\[
		\norm{P_x-\beta sP_x}
		=\abs{1-\beta s}\norm{P_x}
		\leq\epsilon.
		\]
		Thus every $U_x$ is a $(\beta,b,\epsilon)$-block-encoding
		of $P_x$, with the same supplied parameters.
		Since $P_{\mathbf0}=0$ and $P_{e_j}\neq0$, this also proves
		the converse in item~\ref{item:idempotent-admissibility}.
		
		Finally, tensoring with an identity does not change the
		operator norm, so
		\begin{align*}
			\eta^2
			&:=\norm{V_1-V_0}^2\\
			&=\norm{R(s)-R(0)}^2\\
			&=s^2+\bigl(\sqrt{1-s^2}-1\bigr)^2\\
			&=\frac{2s^2}{1+\sqrt{1-s^2}}
			\leq2s^2.
		\end{align*}
		Deciding whether $P_x=0$ distinguishes $x=\mathbf0$
		from $x\in\{e_1,\ldots,e_n\}$. All known parameters are
		fixed across this family, so \cref{thm:search-hybrid}
		applies. Since $0<\eta\leq\sqrt2\,s$, it gives
		\[
		T=\Omega\!\left(\frac{\sqrt n}{s}\right)
		=\Omega\!\left(\max\{1,\beta\}\sqrt n\right)
		=\Omega(\beta\sqrt n).
		\]
		This proves item~\ref{item:idempotent-query-complexity}.
	\end{proof}
	In general, $\norm{P}$, and hence any admissible normalization $\beta$ in the above theorem, can be arbitrarily large even for $n=2$. 
	For example, consider a matrix $P$ similar to $\ket 0\bra 0$, but with nearly collinear eigenvectors: $v_1=\ket{0}$ and $v_2=(1-\gamma)\ket{0} +\gamma\ket{1}$, for $\gamma>0$. 
	Then, for $Pv_1=v_1$ and $Pv_2=0$, 
	\[
		\norm{P}=\frac{\sqrt{\gamma^2+(1-\gamma)^2}}{\gamma}=\Theta\round*{\frac 1 \gamma}
	\] 
	as $\gamma \to 0^+$.

    In the next result, we provide an algorithm which solves the zero testing idempotent problem of \cref{thm:idempotent-zero-lower-bound}, when  a sufficiently accurate block encoding of the idempotent is given. Under this assumption, the algorithm achieves the lower bound from \cref{thm:idempotent-zero-lower-bound}. 
    
    \begin{theorem}\label{thm:zero_test_idempotent}
		Let $\delta \in (0,0.5)$.
		Suppose that $U_{P}$ is a $(\beta,b,\epsilon_+)$-block-encoding of  an idempotent matrix  $P\in \mathbb{C}^{n\times n}$, where $\epsilon_+ \leq \frac{1}{4\sqrt{n}}$.
		There is a quantum algorithm that determines  whether $P$ is zero with success probability at least $1-\delta$
    	on every admissible input. Let $Q$ denote its worst-case
    	number of queries to $U_P$ and $U_P^\dagger$. Then
    	\[
    		Q=O\!\left(\beta\sqrt n\log(1/\delta)\right).
    	\]

		Moreover, for fixed $\delta$, this algorithm has optimal
		worst-case query complexity up to constant factors, with
		\[
		Q=
		\begin{cases}
			\Theta(\beta\sqrt n)
			& \beta+\epsilon_+\geq1 \text{ and } b\geq 1 \\
			0
			& \text{otherwise}.
		\end{cases}
		\]
	\end{theorem}

    \begin{proof}
		Let $\tilde{P}= (\langle {0^b} \rvert\otimes I)U_P(\lvert {0^b} \rangle\otimes I)$, so that $\norm{P-\beta \tilde{P}}\leq \epsilon_+$. 
		
		The quantum algorithm proceeds like so:
		\begin{enumerate}
			\item
			If $\beta+\epsilon_{+} <1$ output $P$ is zero and halt. 
			If $b=0$ and $\beta+\epsilon_{+} \geq 1$ output $P$ is not zero and halt. 
			\item Prepare
			\[
			\lvert {0^b} \rangle\ket{\phi},
			\qquad
			\ket{\phi}
			=
			\frac1{\sqrt n}\sum_{j=1}^n\ket j\ket j.
			\]
			\item Apply \(U_P\) to the block-encoding ancillas and the first
			system register, obtaining the normalized state
			\[
			\ket{\Psi}
			=
			(U_P\otimes I_{n})\lvert {0^b} \rangle\ket{\phi}.
			\]
			\item Use amplitude estimation with the good subspace defined by the
			block-encoding ancillas being in \(\lvert {0^b} \rangle\), and denote the
			resulting estimate by \(\tilde p\).
			\item Carry out steps 2--4 a total of $r$ independent
			times, with $r$ odd and chosen below. Denote the median
			of the estimates by $\hat p_m$.
			\item Let $T= \frac{9-4\sqrt{2}}{16\beta^2n}$. 
			If $\hat p_m> T$ then, output $P$ is not zero, otherwise output zero. 
		\end{enumerate}

		Note $P^2=P$.  
		By the final assertion of item 2 of \cref{thm:idempotent-zero-lower-bound}, $\beta+\epsilon_{+}<1$ implies $P=0$, and $b=0$ together with $\beta+\epsilon_{+}\geq1$ implies $P=I_n$. 
		These cases are caught by the first step of the algorithm and require no queries to $U_P$ and $U_P^{\dagger}$, so $Q=0$ in these cases. 
		Therefore, we assume $\beta+\epsilon_+\geq1 \text{ and } b\geq 1$ for the remainder of the proof.

		Projecting $\ket{\Psi}$ into the good subspace gives
		\[
		\round*{\lvert {0^b} \rangle\langle {0^b} \rvert\otimes I_{n^2}}\ket{\Psi}=\lvert {0^b} \rangle\otimes\round*{\tilde{P}\otimes I_n}\ket{\phi}.
		\]
		The probability of measuring the good subspace, that is ancilla qubits corresponding to $\lvert {0^b} \rangle$, is
		\[
			p=\norm*{\lvert {0^b} \rangle\otimes\round*{\tilde{P}\otimes I_n}\ket{\phi}}^2=\norm*{\round*{\tilde{P}\otimes I_n}\ket{\phi}}^2=\frac{1}{ n}\opn{tr}\round*{\tilde{P}^{\dagger}\tilde{P}}=\frac{1}{{n}}\norm{\tilde{P}}_F^2
		\]
		
		Now we bound the quantity $\norm{P}_F^2$. 
        Since $P$ is idempotent it is diagonalizable with eigenvalues $0$ or $1$. 
        Thus, the trace of $P$ corresponds to its rank. 
        Using the trace inner product for matrices and the fact $P^2=P$ we see that
        \[
        	\opn{rank}(P)= \opn{tr}(P)=\opn{tr}(P^{2})=\opn{tr}((P^\dagger)^{\dagger}P) \leq \norm{P^{\dagger}}_F\norm{P}_F=\norm{P}_F^2
        \]
        
		Now consider the case where $P$ is not zero then, $\opn{rank}(P)\geq 1$.
		Since $\tilde{P}=\frac{1}{\beta}P+\frac{1}{\beta}\round*{\beta \tilde{P} - P}$, we can see that

		\begin{align*}
			\sqrt{p}=\norm{(\tilde{P}\otimes I_n)\ket{\phi}} & \geq \frac 1 \beta \norm{(P\otimes I_n)\ket{\phi}}- \frac{1}{\beta}\norm{\round{(\beta \tilde{P} -P)\otimes I_n}\ket \phi}\\
			&\geq \frac 1 \beta \frac{\norm{P}_F}{\sqrt{n}}- \frac{1}{\beta}\norm{\beta \tilde{P} - P}\norm{\ket \phi}\\
			&\geq \frac 1 \beta \frac{\sqrt{\opn{rank}(P)}}{\sqrt{n}} - \frac 1 \beta \epsilon_+\\
			& \geq \frac{1}{\beta}\round*{\frac{1}{\sqrt{n}}-\frac{1}{4\sqrt{n}}}=\frac{3}{4\beta\sqrt n}.
		\end{align*}

		In the case where $P$ is zero,  and since
		$\tilde{P}=\frac{1}{\beta}P+\frac{1}{\beta}\round*{\beta \tilde{P} - P}$ we have 
		$\norm{\tilde{P}}\leq \frac{\epsilon_+}{\beta}$. Hence,
		\begin{align*}
			\sqrt{p}=\norm{(\tilde{P}\otimes I_n)\ket{\phi}}&\leq \norm{\tilde{P}\otimes I_n}\norm{\ket{\phi}} \\
			&= \norm{\tilde{P}}\leq \frac{\epsilon_+}{\beta} \leq \frac{1}{4\beta\sqrt{n}}. 
		\end{align*}
		
		Thus the amplitude \(\sqrt p\) is either at most
		\[
		a_0:=\frac1{4\beta\sqrt n}
		\]
		or at least
		\[
		a_1:=\frac3{4\beta\sqrt n}=3a_0.
		\]
		
		By using amplitude estimation~\cite[Theorem~12]{BrassardGilles2000QAAa}, we can get an estimate of $ {p}$, say $ {\tilde p}$, which satisfies
		\begin{equation}\label{eq:amp_bound}
			\abs{ {\tilde p}- {p}}\leq 2\pi \frac{ \sqrt{ {p}(1- {p})}}{M}+\frac{\pi^2}{M^2}
		\end{equation}
		with probability greater than $\frac{8}{\pi^2}$ using $O(M)$ applications of \(U_P\)
		and \(U_P^\dagger\). 
		For the following deductions, suppose that \cref{eq:amp_bound} holds.
		Let $\mu = 2(\sqrt{2}-1)a_0$ and pick $M$ to be the least integer with $M>{ \pi\mu^{-1}}$ (so $M=O(\beta \sqrt{n})$) then note $\sqrt{ {p}(1- {p})}\leq \sqrt{p}$. We can see that 
		\begin{align*}
			\abs{ {\tilde p}- {p}}< 2\frac{\pi}{\pi\mu^{-1}}\sqrt p +\round*{\frac{\pi}{\pi\mu^{-1}} }^2 &= - p + \round*{\sqrt p + {\mu}}^2\\
			&=+p-\round*{\sqrt p - {\mu}}^2 +  2 {\mu^2}.
		\end{align*}
		Using these bounds one finds that
		\begin{equation}\label{eq:estimate_bound}
			\round*{\sqrt p - {\mu}}^2 -  2 {\mu^2} < \tilde{p} < \round*{\sqrt p + {\mu}}^2.
		\end{equation}
		
		Let $T=(9-4\sqrt 2)a_0^2$.

		If $P$ is not zero,  $\sqrt p\geq a_1 \geq a_0>\mu$ and using the lower bound in \cref{eq:estimate_bound} then,
		\[
			T=(9-4\sqrt 2)a_0^2=\round*{a_1 - {\mu}}^2 -  2 {\mu}^2 \leq \round*{\sqrt p - {\mu}}^2 -  2 {\mu^2} <\tilde{p}.
		\]
		Taking the contrapositive, we get $\tilde{p} \leq T$ implies $P$ is zero. 
		
		If $P$ is zero, $\sqrt p\leq a_0$ and using the upper bound yields,
		\[
			 \tilde{p}<  \round*{\sqrt p + {\mu}}^2\leq  \round*{a_0 + {\mu}}^2=(9-4\sqrt 2)a_0^2=T
		\]
		It follows that $\tilde{p} \geq T$ implies $P$ is not zero.

		Specifically with probability at least $\frac 8 {\pi^2}$, if $\tilde{p} \leq T$ then $P$ must be zero and if $\tilde{p} >T$ then $P$ is not zero. 
		
		Run this estimation procedure independently an odd number
		$r$ of times and take the median of the resulting estimates.
		Return \emph{not zero} if the median exceeds $T$,
		and \emph{zero} otherwise. The decision is correct whenever
		more than half of the estimates satisfy
		\cref{eq:amp_bound}. Hoeffding's inequality
		\cite{Hoeffding1963} therefore gives
		\[
		\Pr(\text{incorrect decision})
		\leq
		\exp\!\left[
		-2r\left(\frac8{\pi^2}-\frac12\right)^2
		\right].
		\]
		Choose the smallest odd $r$ for which this bound is at most
		$\delta$. Then $r=O(\log(1/\delta))$, and the total query
		count is
		\[
		O(Mr)
		=O\!\left(\beta\sqrt n\log(1/\delta)\right).
		\]

		For the lower bound, we apply item 2 of
		\cref{thm:idempotent-zero-lower-bound}.

	\end{proof}
      We can now apply the preceding result to decide Hurwitz stability, given a
block encoding of the spectral projector.
    \begin{corollary}\label{thm:determine_stable_encoding_supplied}
    	Suppose that no eigenvalue of \(A\in\mathbb{C}^{n\times n}\) is purely imaginary. Let \(\delta\in(0,0.5)\). Suppose that \(U_P\) is a
        \((\beta,b,\epsilon_+)\)-block-encoding of \(\operatorname{sign}^+(A)\),
        where \(\epsilon_+\leq \frac{1}{4\sqrt{n}}\).
        
        There is a quantum algorithm that determines whether \(A\) is Hurwitz with
        success probability at least \(1-\delta\). 
        Let \(Q\) denote its worst-case number of queries to \(U_P\) and \(U_P^\dagger\).
        Then
        \[
            Q=O\!\left(\beta\sqrt{n}\log(1/\delta)\right).
        \]
        
        Moreover, for fixed \(\delta\), this algorithm has optimal worst-case query
        complexity up to constant factors, with
        \[
        Q=
        \begin{cases}
        \Theta(\beta\sqrt{n})
            & \beta+\epsilon_+\geq 1 \text{ and } b\geq 1,\\
        0
            & \text{otherwise}.
        \end{cases}
        \]
    \end{corollary}
    \begin{proof}
         Apply the algorithm in \cref{thm:zero_test_idempotent} to $P=\operatorname{sign}^+(A)$, and output Hurwitz if and only if it determines that $P=0$.

    \end{proof}

    We discuss how $\beta$ might depend on $\omega_A$ (or not) in \cref{sec:beta_dependence_on_omega}.

\section{Determining stability using the block-encoding of $A$} \label{sec:determining_stability}

    \Cref{thm:determine_stable_encoding_supplied} assumes access to the block-encoding of $\sign^+(A)$, rather than of $A$. 
    Our algorithm proceeds in two steps: it first constructs a block-encoding of $\sign^+(A)$ from $U_A$, and then applies the zero test. \Cref{thm:determine_stable_encoding_supplied} gives the optimal complexity of the second step; it remains to account for the first. 
    In \cref{sec:matrix_sign} we introduce a method of producing an accurate block-encoding of $\sign^+(A)$, which culminates with \cref{thm:sign_complexity}. 
    The results of this section are simple and direct applications of \cref{thm:determine_stable_encoding_supplied} and  \cref{thm:sign_complexity}.

	\begin{theorem}\label{thm:determine_stable_encoding_ours}
		Suppose that $A\in \C^{n\times n}$ has $\norm{A}\leq 1$ and no eigenvalue of $A$ is purely imaginary. 
		Assume that $\nu\in (0,\omega_A]$ is known and let $\delta \in (0,0.5)$. 
		Let $U_A$ be a $(1,a,0)$-block-encoding of $A$.

        There is a quantum algorithm that decides whether $A$ is Hurwitz, with success probability at least $1-\delta$, using
		\[
			O\round*{\frac{\sqrt{n}}{\nu^2}\log\round*{{\frac{\sqrt{n}}{\nu}}}\log\round*{\frac 1 \delta}}
		\]
		calls to $U_A$ and $U_A^\dagger$.
		
	\end{theorem}
	\begin{proof}
        Set $\epsilon_+=1/(4\sqrt n)$. 
		If \cref{thm:sign_complexity} is used to produce the block encoding of $P=\sign^+(A)$ with error $\epsilon_+$ and $\nu \leq \omega_A$ is known then, $\beta =O(\nu^{-1})$ and $\sign^+(A)$ uses $O(\nu^{-1}\log(\nu^{-1} \epsilon_+^{-1}))=O(\nu^{-1}\log(\sqrt n / \nu) )$ calls to the block encoding of $A$. 
		Applying the amplitude estimation procedure from \cref{thm:determine_stable_encoding_supplied} would use $O(\beta\sqrt{n}\log(\delta^{-1}))$ calls to $U_P$ and $U_P^\dagger$. 
		Therefore, the total complexity is
		\[
			O\round*{\frac{\sqrt{n}}{\nu}\log\round*{\frac 1 \delta} \times \frac 1 \nu \log\round*{{\frac{\sqrt{n}}{\nu}}}}
			=
			O\round*{\frac{\sqrt{n}}{\nu^2}\log\round*{{\frac{\sqrt{n}}{\nu}}}\log\round*{\frac 1 \delta}}.
		\]
	\end{proof}

	\begin{theorem}[When $A$ is Hermitian or normal] \label{thm:determine_stable_encoding_normal}
		Suppose that $A\in \C^{n\times n}$ has $\norm{A}\leq 1$ and no eigenvalue of $A$ is purely imaginary. 
		Let $\delta \in (0,0.5)$. 
		Assume that $\nu\in (0,\omega_A]$ is known and that $A$ is Hermitian or normal.
		Let $U_A$ be a $(1,a,0)$-block-encoding of $A$.
		
		For every $\epsilon_+\in(0,1]$, one can construct, using $O(\nu^{-1}\log(1/\epsilon_+))$ queries to $U_A$ and $U_A^\dagger$, a $(1,a+3,\epsilon_+)$-block-encoding of $\sign^+(A)$.
		When $\epsilon_+ = \frac{1}{4\sqrt{n}}$, the algorithm from \cref{thm:determine_stable_encoding_supplied} uses 
		\[
			O\round*{\frac{\sqrt{n}}{\nu}\log\round*{{n}}\log\round*{\frac 1 \delta}}
		\]
		calls to $U_A$ and $U_A^\dagger$ and succeeds with probability at least $1-\delta$.
	\end{theorem}
	\begin{proof}
        Although the theorem assumes an exact encoding, the following argument also covers inexact encodings of normal matrices, a case not covered by \cref{thm:omega_A} alone. 
        As a technical detail, just because we know that $A$ is normal/Hermitian doesn't mean that 
        $\widetilde A_U=(\bra{0^a}\otimes I)U(\ket{0^a}\otimes I)$ is normal/Hermitian. 
        So we first argue that we can WLOG assume that $\widetilde A_U$ is a Hermitian matrix. 
    
        Let $U$ be an $(\alpha,a,\epsilon)$-block-encoding of a normal matrix $A\in\C^{n\times n}$, with a known $\epsilon<\nu\leq\omega_A$. 
        Let $H(A)=0.5(A+A^{\dagger})$ be the Hermitian part of $A$. 
        Using LCU to approximately construct $H(A)$ gives us a $(1,a+1,0)$-block-encoding of
        $H(\widetilde A_U)$, using one controlled query each to $U$ and
        $U^\dagger$. 
        Since $A$ is normal,
        \[
        \omega_{H(A)}=\omega_A,\qquad
        \sign[H(A)]=\sign(A),\qquad
        \norm{H(A)-\alpha H(\widetilde A_U)}\leq\epsilon.
        \]
        Proposition~\ref{thm:omega_A} therefore gives
        \[
        \norm{H(\widetilde A_U)}\leq1,\qquad
        \omega_{H(\widetilde A_U)}\geq\frac{\nu-\epsilon}{\alpha}:=\tilde{\nu},
        \qquad
        A\text{ is Hurwitz}\iff H(\widetilde A_U)\text{ is Hurwitz}.
        \]
        Since $H(\widetilde A_U)$ is Hermitian, we can assume below that we are given access to a $(1,a+1,0)$-block-encoding of a Hermitian matrix. 
        The lower bound becomes $\tilde{\nu}$. 
        When $U_A$ is a $(1,a,0)$-block-encoding of $A$ (as the statement of this theorem assumes) then, $\nu=\tilde{\nu}$. 
        Therefore, even when a normal/Hermitian $A$ is inexactly encoded (but still accurate enough, again see \cref{thm:omega_A}) we can still assume that $A$ is Hermitian and exactly encoded.

		We use~\cite[Lemma 25 \& Theorem 26]{QSVT} to form a $(1,(a+1)+1,\epsilon_+)$-block-encoding of $\sign[A]$. 
		We can do this by picking their $\delta=\nu$.
		Applying the two-term LCU construction~\cite[Lemma~52]{QSVT} to $I$ and this block-encoding of $\sign(A)$, with coefficients $1/2,1/2$, gives a $(1,a+3,\epsilon_+/2)$-block-encoding of $\sign^+(A)$. It uses one controlled application of the sign block-encoding and two Hadamard gates. 
		By \cref{thm:determine_stable_encoding_supplied} we get
		\[
			O\round*{\sqrt{n} \log\round*{\frac 1 \delta} \times \frac 1 \nu \log\round*{{\sqrt{n}}}}
			=O\round*{\frac{\sqrt{n}}{\nu}\log\round*{{n}}\log\round*{\frac 1 \delta}}
		\]
		calls to $U_A$ and $U_A^\dagger$.

    \end{proof}

    The improvement in \cref{thm:determine_stable_encoding_normal} (normal $A$) over \cref{thm:determine_stable_encoding_ours} (general $A$) is structural rather than a matter of implementation. For normal $A$ with $\sign^+(A)\neq0$, the projector $\sign^+(A)$ is orthogonal, so $\norm{\sign^+(A)}=1$ and an exact encoding with $\beta=1$ is possible. For non-normal $A$, $\sign^+(A)$ may be an oblique projector with $\norm{\sign^+(A)}>1$, and every exact $(\beta,b,0)$-block-encoding then has $\beta\geq\norm{\sign^+(A)}>1$. 
    Worst case scaling for $\norm{\sign^{+}(A)}$ in terms of $\omega_A$ is given in \cref{sec:beta_dependence_on_omega}. 

    \begin{remark}[Possible improvement for normal matrices] \label{rem:normal_eigenphase_extension}
        The factor $\log n$ in \cref{thm:determine_stable_encoding_normal} may be removable using the eigenphase-threshold algorithm of \cite[Corollary~4.1]{JefferyWitteveen2026GroundState}.
        For normal $A$, the Hermitian matrix $H=(A+A^\dagger)/2$ has eigenvalues equal to the real parts of those of $A$, so $A$ is Hurwitz if and only if $\lambda_{\max}(H)<0$.
        At constant failure probability, their algorithm uses $O((\gamma\Delta)^{-1})$ unitary queries, where $\Delta$ is the promised phase separation and $\gamma$ is the amplitude overlap with the relevant spectral subspace.
        A maximally entangled input supplies amplitude at least $n^{-1/2}$ in every nonzero eigenspace of $H\otimes I_n$.
        A suitable reduction preserving this overlap up to a constant factor, with phase separation $\Theta(\nu)$ and $O(1)$ queries to $U_A$ per unitary query, would therefore yield an $O(\sqrt n/\nu)$ decision bound, matching \cref{thm:diagonal-Hurwitz-complexity}.
        Constructing this reduction in our access model, including satisfying their eigenphase-window promise, requires a separate algorithmic development, which we do not pursue here.
    \end{remark}

\section{Parameter dependence and extensions}
\label{sec:Discussion}

\subsection{$\omega_A$ and eigenvector conditioning}\label{sec:omega_eigenvector_conditioning}
For a diagonalizable matrix $A=VDV^{-1}$, the columns of $V$ form an eigenvector basis and the diagonal entries of $D$ are the corresponding eigenvalues.
The condition number $\kappa_2(V)=\norm{V}\norm{V^{-1}}$ measures the sensitivity of this change of basis to perturbations and enters the eigenvalue-estimation query bounds of \cite{ZhangEtAl2025NonHermitian,LinShang2026NonNormal}.

To see why $\omega_A$ does not control this condition number, consider, for $0<\varepsilon\leq1$,
\[
    A_\varepsilon=\frac14
    \begin{pmatrix}
        -2&1\\
        0&-2-\varepsilon
    \end{pmatrix},
    \qquad
    \norm{A_\varepsilon}\leq\norm{A_\varepsilon}_F\leq\frac{\sqrt{14}}4<1.
\]
Its Hermitian part satisfies $(A_\varepsilon+A_\varepsilon^\dagger)/2\preceq-3I/8$, which gives $\sigma_{\min}(A_\varepsilon-\mathrm{i}tI)\geq3/8$ for every $t\in\R$ and hence $\omega_{A_\varepsilon}\geq3/8$.
The eigenvalues $-1/2$ and $-(2+\varepsilon)/4$ have respective eigendirections $(1,0)^\top$ and $(1,-\varepsilon)^\top$, which approach one another as $\varepsilon\downarrow0$.
Let $V_\varepsilon$ have any nonzero eigenvectors in these directions as its columns, in that order.
The projector onto the first eigenspace along the second is
\[
    V_\varepsilon
    \begin{pmatrix}1&0\\0&0\end{pmatrix}
    V_\varepsilon^{-1}
    =\begin{pmatrix}1&\varepsilon^{-1}\\0&0\end{pmatrix}.
\]
Its norm is at most $\kappa_2(V_\varepsilon)$, so
\[
    \kappa_2(V_\varepsilon)\geq\sqrt{1+\varepsilon^{-2}}\longrightarrow\infty
    \qquad(\varepsilon\downarrow0).
\]
Thus every eigenvector basis becomes ill-conditioned, regardless of column scaling, while $n=2$ and the margin lower bound $\nu=3/8$ remain fixed.

\subsection{Dependence of $\omega_{A}$ on dimension $n$}
	The bounds above depend on $\omega_A^{-1}$, which raises the question of how $\omega_A$ behaves as the dimension grows. Posed for arbitrary matrices the question has no answer, but for a random family it does. 
	\begin{proposition}(See the proof in \cref{sec:Proof_thm:rand_diag_matrix}) \label{thm:rand_diag_matrix}
		Let $n\in \N$ have $n\geq 2$. 
		Consider real diagonal matrices $A=\opn{diag}(a_j)_{j=1}^n$ where each $a_j$ is chosen independently and uniformly at random from the interval $[-1,1]$. 
		Then, 
		    \[
		\Pr\left(
		\frac{1}{5n}\leq\omega_A\leq\frac2n
		\right)
		=
		\left(1-\frac{1}{5n}\right)^n
		-
		\left(1-\frac2n\right)^n
		>0.67.
		\]
	\end{proposition}
	
	Thus, at least 67\% of matrices in this random family satisfy
	$1/(5n)\leq\omega_A\leq2/n$. Moreover,
	$\mathbb {E}[1/\omega_A]=\infty$, so an average-case analysis based
	directly on this reciprocal is not informative.
	Taking $\nu=\Theta(n^{-1})$ in the oracle family of, e.g., 
	\cref{thm:diagonal-Hurwitz-complexity} gives a worst-case lower bound
	of $\Omega(n^{3/2})$ queries. 
    This is worse than the classical cost of deciding whether a matrix in this class is Hurwitz, which is $\Theta(n)$ sign-bit checks in the worst case. 
	This apparent contradiction is an artifact of the access model in \cref{thm:diagonal-Hurwitz-complexity}; 
	oracle access to the block-encoding of $U_A$ and $U^\dagger_A$ is not classically analogous to having access to the sign bits. 
\subsection{Dependence of the normalization factor $\beta$}\label{sec:beta_dependence_on_omega}
  In \cref{thm:zero_test_idempotent} and \cref{thm:determine_stable_encoding_supplied}, the query complexities depend on the normalization or scaling parameter $\beta$. 
  The scaling parameter $\beta$ (which is often proportional to $\norm{\sign^+(A)}$)  may grow without bound as $\omega_A\to0$, even for \(n=2\) and \(\norm{A}=1\). 
    For example, let \(0<\gamma\leq1\) and consider
    \[
        A=\begin{pmatrix}
        0&1\\
        \gamma&0
        \end{pmatrix}.
    \]
    Then \(\norm{A}=1\), and \(A\) has eigenvalues \(\pm\sqrt{\gamma}\), so
    \[
        \sign^+(A)=\frac12\left(I+\frac{A}{\sqrt{\gamma}}\right)
        =\frac12
        \begin{pmatrix}
            1&\gamma^{-1/2}\\
            \gamma^{1/2}&1
        \end{pmatrix}.
    \]
    One can show that $\omega_A=\gamma$. 
    One can also see that
    \[
        \norm{\sign^+(A)}
        =\frac{1+\gamma}{2\sqrt{\gamma}}
        =\frac{1+\omega_A}{2\sqrt{\omega_A}}
        =\Theta\!\left(\omega_A^{-1/2}\right)
        \qquad\text{as }\gamma\to0^+.
    \]

    On the other hand, $\beta$ (and $\norm{\sign^+(A)}$) can also be independent of $\omega_A$. 
    When $-A$ is Hurwitz, $\sign^+(A)=I$. 
    So a $(1,0,0)$-block-encoding of $\sign^+(A)$ is available when $A=\omega I$ for all $\omega>0$. 

    In the next result, whose proof is given in \cref{sec:Proof_thm:fixed_dimensional_normalization}, we bound the norm of the projector in terms of $\omega_A$ for fixed dimension $n$. We then explore the ramifications of such bounds on the previously derived complexity bounds. 
		\begin{theorem}\label{thm:fixed_dimensional_normalization}
		
		Fix \(n\geq2\). 
		For every matrix $A\in \C^{n\times n}$ with $\norm{A}\leq 1$ and no purely imaginary eigenvalues we have 
		\[
		\norm{\sign^+(A)}
		=
		O\left(\omega_A^{-(n-1)/n}\right). 
		\]
		Moreover, there are such matrices with $\norm{\sign^+(A)}= \Theta\left(\omega_A^{-(n-1)/n}\right)$.

	\end{theorem}

	Now, let us fix $n\geq2$. Let $A\in \C^{n\times n}$ have $\norm{A}\leq 1$, no purely imaginary eigenvalues and let $\nu \in (0,\omega_{A}]$. 
	Let $\delta \in (0,1/2)$. 
	Let $U_P$ be a $(\beta,b,0)$-block encoding of $\sign^+(A)$, with $\beta=\max\{1,\norm{\sign^{+}(A)}\}$. 
    It follows that $\beta=\norm{\sign^+(A)}$ whenever $\sign^+(A)\neq0$.
	
	A consequence of \cref{thm:fixed_dimensional_normalization} is that the query bound supplied by \cref{thm:determine_stable_encoding_supplied} is at worst
	\[
    	O\left(
    	\sqrt n\,\nu^{-(n-1)/n}
    	\log\left(\frac1\delta\right)
    	\right)
	\]
	queries to $U_P$ (the block encoding of $\sign^{+}(A)$) and $U_P^\dagger$. Note that,  herein, we do not account for the query complexity associated with computing the projector $\sign^{+}(A)$.
    That being said, \Cref{thm:fixed_dimensional_normalization} does suggest that $(O(\nu^{-(n-1)/n}),b,\epsilon_{+})$-block-encodings of $\sign^+(A)$ are possible as opposed to the $(O(\nu^{-1}),b,\epsilon_{+})$-block-encoding supplied by \cref{thm:sign_complexity}.

\section{Conclusions and future work}\label{sec:conclusion}

We have established upper and lower query bounds for deciding
Hurwitz stability from a block-encoding and a known lower bound
$0<\nu\leq\omega_A$. For constant success probability, the
worst-case lower bound is $\Omega(\sqrt n/\nu)$, already for
real diagonal matrices. Our algorithm matches this bound up to
logarithmic factors for normal matrices and uses
$\widetilde O(\sqrt n/\nu^2)$ queries for general matrices.
For a supplied encoding of $\sign^+(A)$ with normalization
$\beta$, zero testing has optimal query complexity
$\Theta(\beta\sqrt n)$ under the stated accuracy condition,
whenever both decision outcomes are admissible.

The main open question is whether the additional factor $1/\nu$
in the general-matrix upper bound is necessary. Our sharp
fixed-dimensional bounds on $\norm{\sign^+(A)}$ quantify the
normalization required by exact projector encodings. Whether
these bounds can be exploited in a more efficient construction,
or whether a direct stability test can improve the query
complexity, remains open. Stronger lower bounds for non-normal
inputs would also help determine the optimal dependence on
the stability margin.

\section*{Author contributions}

 M.F. was responsible for the writing, the main premise, and
all mathematical details. 
M.S.\ verified and refined several mathematical details, helped organize the presentation of the results, 
and contributed to the writing and revision of the manuscript.
M.K. contributed to the conception and proposal of the project as co-Principal Investigator, provided supervision, and reviewed the manuscript. 
J.L. conceptualized the project, developed preliminary versions of
\cref{thm:determine_stable_encoding_supplied} and
\cref{thm:diagonal-Hurwitz-complexity}, and contributed to the
writing of the manuscript. All authors reviewed and approved the
final manuscript.

ChatGPT (OpenAI)  and Claude (Anthropic)  were used to assist with drafting and revising
text, LaTeX formatting, literature searches, and the development
and checking of mathematical arguments. The authors take full
responsibility for the final manuscript, including its references
and mathematical results.

\section*{Acknowledgments}
 This work was supported by the National Research
Council of Canada (NRC) through the Applied Quantum Computing
(AQC) Challenge program under project \mbox{AQC-318}. J.L. also acknowledges
the support of the Natural Sciences and Engineering Research
Council of Canada (NSERC) and the Canada Research Chairs
program.

\section*{Data availability}
No new data were created or analysed in this study.
\bibliographystyle{quantum}
\bibliography{biblio_quantum}

@article{KerznerEtAl2024SmallestEigenvalue,
  author  = {Kerzner, Alex and Gheorghiu, Vlad and Mosca, Michele and Guilbaud, Thomas and Carminati, Federico and Fracas, Fabio and Dellantonio, Luca},
  title   = {A Square-Root Speedup for Finding the Smallest Eigenvalue},
  journal = {Quantum Science and Technology},
  volume  = {9},
  number  = {4},
  pages   = {045025},
  year    = {2024},
  doi     = {10.1088/2058-9565/ad6a36},
  url     = {https://doi.org/10.1088/2058-9565/ad6a36}
}

@book{AstromMurray2021,
  author    = {{\AA}str{\"o}m, Karl J. and Murray, Richard M.},
  title     = {Feedback Systems: An Introduction for Scientists and Engineers},
  edition   = {2nd},
  publisher = {Princeton University Press},
  address   = {Princeton, NJ},
  year      = {2021},
  isbn      = {9780691193984},
  url       = {https://press.princeton.edu/books/hardcover/9780691193984/feedback-systems}
}

@article{HinrichsenPritchard1986,
  author  = {Hinrichsen, Diederich and Pritchard, Anthony J.},
  title   = {Stability Radii of Linear Systems},
  journal = {Systems \& Control Letters},
  volume  = {7},
  number  = {1},
  pages   = {1--10},
  year    = {1986},
  doi     = {10.1016/0167-6911(86)90094-0},
  url     = {https://doi.org/10.1016/0167-6911(86)90094-0}
}

@article{TrefethenEtAl1993,
  author  = {Trefethen, Lloyd N. and Trefethen, Anne E. and Reddy, Satish C. and Driscoll, Tobin A.},
  title   = {Hydrodynamic Stability Without Eigenvalues},
  journal = {Science},
  volume  = {261},
  number  = {5121},
  pages   = {578--584},
  year    = {1993},
  doi     = {10.1126/science.261.5121.578},
  url     = {https://doi.org/10.1126/science.261.5121.578}
}

@article{JenningsEtAl2024,
  author  = {Jennings, David and Lostaglio, Matteo and Lowrie, Robert B. and Pallister, Sam and Sornborger, Andrew T.},
  title   = {The Cost of Solving Linear Differential Equations on a Quantum Computer: Fast-Forwarding to Explicit Resource Counts},
  journal = {Quantum},
  volume  = {8},
  pages   = {1553},
  year    = {2024},
  doi     = {10.22331/q-2024-12-10-1553},
  url     = {https://doi.org/10.22331/q-2024-12-10-1553}
}

@article{LowSu2024QEVT,
  author  = {Low, Guang Hao and Su, Yuan},
  title   = {Quantum Eigenvalue Processing},
  journal = {SIAM Journal on Computing},
  volume  = {55},
  number  = {1},
  pages   = {135--215},
  year    = {2026},
  doi     = {10.1137/24M1689363},
  url     = {https://doi.org/10.1137/24M1689363}
}

@book{higgs09,
author = {Higham, Nicholas J.},
title = {Functions of Matrices: Theory and Computation},
publisher = {Society for Industrial and Applied Mathematics},
year = {2008},
doi = {10.1137/1.9780898717778},
address = {},
edition   = {},
URL = {https://epubs.siam.org/doi/abs/10.1137/1.9780898717778}
}

@article{Kenny91,
author = {Kenney, Charles and Laub, Alan J.},
title = {Rational Iterative Methods for the Matrix Sign Function},
journal = {SIAM Journal on Matrix Analysis and Applications},
volume = {12},
number = {2},
pages = {273-291},
year = {1991},
doi = {10.1137/0612020},
URL = {https://doi.org/10.1137/0612020}
}

@article{TakahiraOhashiSogabeUsuda2022ContourBlockEncoding,
  author  = {Takahira, Souichi and Ohashi, Asuka and Sogabe, Tomohiro and Usuda, Tsuyoshi S.},
  title   = {Quantum Algorithms Based on the Block-Encoding Framework for Matrix Functions by Contour Integrals},
  journal = {Quantum Information and Computation},
  volume  = {22},
  number  = {11\&12},
  pages   = {965--979},
  year    = {2022},
  doi     = {10.26421/QIC22.11-12-4},
  url     = {https://doi.org/10.26421/QIC22.11-12-4}
}

@incollection{BrassardGilles2000QAAa,
  author    = {Brassard, Gilles and H{\o}yer, Peter and
               Mosca, Michele and Tapp, Alain},
  title     = {Quantum Amplitude Amplification and Estimation},
  booktitle = {Quantum Computation and Information},
  series    = {Contemporary Mathematics},
  volume    = {305},
  pages     = {53--74},
  publisher = {American Mathematical Society},
  year      = {2002},
  doi       = {10.1090/conm/305/05215}
}

@article{BennettEtAl1997,
  author  = {Bennett, Charles H. and Bernstein, Ethan and
             Brassard, Gilles and Vazirani, Umesh},
  title   = {Strengths and Weaknesses of Quantum Computing},
  journal = {SIAM Journal on Computing},
  volume  = {26},
  number  = {5},
  pages   = {1510--1523},
  year    = {1997},
  doi     = {10.1137/S0097539796300933}
}

@article{ROBERTS01101980,
author = {Roberts, J. D.},
title = {Linear model reduction and solution of the algebraic {Riccati} equation by use of the sign function},
journal = {International Journal of Control},
volume = {32},
number = {4},
pages = {677--687},
year = {1980},
publisher = {Taylor \& Francis},
doi = {10.1080/00207178008922881},
URL = {https://doi.org/10.1080/00207178008922881}
}

@misc{wang2026signembeddingquantumalgorithms,
      title={Sign Embedding Quantum Algorithms for Matrix Equations and Matrix Functions},
      author={Yanqiao Wang and Jin-Peng Liu},
      year={2026},
      eprint={2604.25333},
      archivePrefix={arXiv},
      primaryClass={quant-ph},
      url={https://arxiv.org/abs/2604.25333},
}

@inproceedings{QSVT,
  eprint = {1806.01838},
  archivePrefix = {arXiv},
  note = {Theorem numbering refers to the full arXiv version},
author = {Gily{\'e}n, Andr{\'a}s and Su, Yuan and Low, Guang Hao and Wiebe, Nathan},
title = {Quantum singular value transformation and beyond: exponential improvements for quantum matrix arithmetics},
year = {2019},
isbn = {9781450367059},
publisher = {Association for Computing Machinery},
address = {New York, NY, USA},
url = {https://doi.org/10.1145/3313276.3316366},
doi = {10.1145/3313276.3316366},
booktitle = {Proceedings of the 51st Annual ACM SIGACT Symposium on Theory of Computing},
pages = {193--204},
numpages = {12},
location = {Phoenix, AZ, USA}
}

@misc{Vandaele26,
      title={Asymptotically Optimal Quantum Circuits for Comparators and Incrementers},
      author={Vivien Vandaele},
      year={2026},
      eprint={2603.12917},
      archivePrefix={arXiv},
      primaryClass={quant-ph},
      url={https://arxiv.org/abs/2603.12917},
}

@misc{frommer2021matrixfunctionslinearsystems,
      title={Matrix functions via linear systems built from continued fractions},
      author={Andreas Frommer and Karsten Kahl and Manuel Tsolakis},
      year={2021},
      eprint={2109.03527},
      archivePrefix={arXiv},
      primaryClass={math.NA},
      url={https://arxiv.org/abs/2109.03527},
}

@article{perez2022linearizations,
title = {Linearizations of rational matrices from general representations},
journal = {Linear Algebra and its Applications},
volume = {647},
pages = {89-126},
year = {2022},
issn = {0024-3795},
doi = {10.1016/j.laa.2022.04.007},
url = {https://www.sciencedirect.com/science/article/pii/S0024379522001525},
author = {Javier Pérez and María C. Quintana}
}

@misc{rodenasruiz2026quantumsolversnonlinearmatrix,
      title={Quantum Solvers for Nonlinear Matrix Equations in Quantum Chemistry},
      author={Pablo Rodenas-Ruiz and Andrew Zhao and Joonho Lee},
      year={2026},
      eprint={2605.16189},
      archivePrefix={arXiv},
      primaryClass={quant-ph},
      url={https://arxiv.org/abs/2605.16189},
}

@misc{nie2024quantumcircuitmultiqubittoffoli,
      title={Quantum circuit for multi-qubit {Toffoli} gate with optimal resource},
      author={Junhong Nie and Wei Zi and Xiaoming Sun},
      year={2024},
      eprint={2402.05053},
      archivePrefix={arXiv},
      primaryClass={quant-ph},
      url={https://arxiv.org/abs/2402.05053},
}

@misc{deWolf2026,
  author        = {de Wolf, Ronald},
  title         = {Quantum Computing: Lecture Notes},
  year          = {2026},
  eprint        = {1907.09415},
  archivePrefix = {arXiv},
  primaryClass  = {quant-ph},
  note          = {Version 6, August 2026},
  url           = {https://arxiv.org/abs/1907.09415v6}
}

@misc{Belovs2015,
  author        = {Belovs, Aleksandrs},
  title         = {Variations on Quantum Adversary},
  year          = {2015},
  eprint        = {1504.06943},
  archivePrefix = {arXiv},
  primaryClass  = {quant-ph},
  url           = {https://arxiv.org/abs/1504.06943}
}

@article{Hoeffding1963,
  author  = {Hoeffding, Wassily},
  title   = {Probability Inequalities for Sums of Bounded Random Variables},
  journal = {Journal of the American Statistical Association},
  volume  = {58},
  number  = {301},
  pages   = {13--30},
  year    = {1963},
  doi     = {10.1080/01621459.1963.10500830}
}

@misc{DalzellSu2026ObliqueEigenprojection,
  author = {Dalzell, Alexander M. and Su, Yuan},
  title = {Quantum oblique eigenprojection},
  year = {2026},
  eprint = {2609.40254},
  archivePrefix = {arXiv},
  primaryClass = {quant-ph},
  url = {https://arxiv.org/abs/2609.40254}
}

@inproceedings{KoladaStefanski2025Stability,
  author = {Kolada, Micha{\l} and Stefa{\'n}ski, Tomasz P.},
  title = {Test of System Stability on Quantum Computer},
  booktitle = {2025 29th International Conference on Methods and Models in Automation and Robotics (MMAR)},
  year = {2025},
  pages = {278--283},
  doi = {10.1109/MMAR65820.2025.11150962},
  url = {https://doi.org/10.1109/MMAR65820.2025.11150962}
}

@misc{LinShang2026NonNormal,
  author = {Lin, Honghong and Shang, Yun},
  title = {Quantum Eigensolver for Non-Normal Matrices via Ground State Energy Estimation},
  year = {2026},
  eprint = {2502.18119},
  archivePrefix = {arXiv},
  primaryClass = {quant-ph},
  note = {Version 3},
  doi = {10.48550/arXiv.2502.18119},
  url = {https://arxiv.org/abs/2502.18119}
}

@misc{ZhangWuYuan2025Eigenvalue,
  author = {Zhang, Yukun and Wu, Yusen and Yuan, Xiao},
  title = {Heisenberg-Limited Quantum Eigenvalue Estimation for Non-normal Matrices},
  year = {2025},
  eprint = {2510.19651},
  archivePrefix = {arXiv},
  primaryClass = {quant-ph},
  doi = {10.48550/arXiv.2510.19651},
  url = {https://arxiv.org/abs/2510.19651}
}

@article{MandeDeWolf2026PhaseEstimation,
  author = {Mande, Nikhil S. and de Wolf, Ronald},
  title = {Tight Bounds for Quantum Phase Estimation and Related Problems},
  journal = {Quantum},
  volume = {10},
  pages = {2140},
  year = {2026},
  doi = {10.22331/q-2026-06-15-2140},
  url = {https://doi.org/10.22331/q-2026-06-15-2140}
}

@article{ZhangEtAl2025NonHermitian,
  author = {Zhang, Xiao-Ming and Zhang, Yukun and He, Wenhao and Yuan, Xiao},
  title = {Exponential Quantum Advantages for Practical Non-{Hermitian} Eigenproblems},
  journal = {Physical Review Letters},
  volume = {135},
  number = {14},
  pages = {140601},
  year = {2025},
  doi = {10.1103/3n8f-k8pl},
  url = {https://doi.org/10.1103/3n8f-k8pl}
}

@misc{JefferyWitteveen2026GroundState,
  author = {Jeffery, Stacey and Witteveen, Freek},
  title = {Optimal Quantum Algorithm for Ground-State Energy Estimation with a Guiding State},
  year = {2026},
  eprint = {2608.24494},
  archivePrefix = {arXiv},
  primaryClass = {quant-ph},
  doi = {10.48550/arXiv.2608.24494},
  url = {https://arxiv.org/abs/2608.24494}
}

@misc{FutamuraYeSakurai2021EigenvalueDensity,
  author = {Futamura, Yasunori and Ye, Xiucai and Sakurai, Tetsuya},
  title = {Contour Integral-based Quantum Algorithm for Estimating Matrix Eigenvalue Density},
  year = {2021},
  eprint = {2112.05395},
  archivePrefix = {arXiv},
  primaryClass = {quant-ph},
  doi = {10.48550/arXiv.2112.05395},
  url = {https://arxiv.org/abs/2112.05395}
}

@book{Weihrauch2000,
  author    = {Weihrauch, Klaus},
  title     = {Computable Analysis: An Introduction},
  series    = {Texts in Theoretical Computer Science},
  publisher = {Springer},
  address   = {Berlin},
  year      = {2000},
  doi       = {10.1007/978-3-642-56999-9}
}
\appendix
\crefalias{section}{appendix}

\section{Proof of \cref{thm:omega_A}}
\label{sec:Proof_thm:omega_A}
		The exact block-encoding assertion follows from the definition.
		
		For item 1, for every \(t\in\R\),
		\begin{align*}
			\alpha\,\sigma_{\min}(\widetilde A_U-itI_n)
			&=\sigma_{\min}(\alpha\widetilde A_U-i\alpha tI_n)\\
			&\geq \sigma_{\min}(A-i\alpha tI_n)
			-\norm{A-\alpha\widetilde A_U}\\
			&\geq \omega_A-\epsilon.
		\end{align*}
		Taking the infimum over \(t\) proves the claim.
		
		For item 2, first recall that \(\omega_B=0\) if and only if
		\(B\) has a purely imaginary eigenvalue. Indeed, the minimum
		defining \(\omega_B\) is attained, since
		\(t\mapsto\sigma_{\min}(B-itI_n)\) is continuous and
		\[
		\sigma_{\min}(B-itI_n)\geq |t|-\norm{B}\longrightarrow\infty
		\qquad\text{as }|t|\longrightarrow\infty.
		\]
		The minimum is zero precisely when \(B-itI_n\) is singular
		for some \(t\in\R\).
		
		Suppose \(\epsilon<\omega_A\), and consider
		\[
		B_s=A+s(\alpha\widetilde A_U-A),\qquad s\in[0,1].
		\]
		The same singular-value estimate gives
		\[
		\omega_{B_s}\geq\omega_A-s\epsilon>0
		\qquad(s\in[0,1]).
		\]
		Thus no \(B_s\) has a purely imaginary eigenvalue. 
        Since \(B_s\) depends continuously on \(s\), its eigenvalues can be followed continuously along this interval, counting algebraic multiplicities. 
        If an eigenvalue moved from one open half-plane to the other, its real part would vanish at some intermediate parameter value. 
        This would give a purely imaginary eigenvalue of \(B_s\), a contradiction.
		Therefore the numbers of eigenvalues in the two open half-planes remain constant along the path.
		
		As \(B_0=A\), \(B_1=\alpha\widetilde A_U\), and \(\alpha>0\), these counts agree for \(A\) and \(\widetilde A_U\).
		Their sign matrices are diagonalizable with the same multiplicities of \(+1\) and \(-1\), so they are similar.
		The Hurwitz equivalence follows.
		
		For item 3, set
		\[
		B=-\omega I_n,\qquad C=\omega I_n,
		\qquad
		V=
		\begin{pmatrix}0&1\\1&0\end{pmatrix}
		\otimes I_{n2^{a-1}}.
		\]
		The matrix \(V\) is unitary and its encoded block is zero.
		Hence
		\[
		\norm{B-\alpha\,0}
		=\norm{C-\alpha\,0}
		=\omega\leq\epsilon.
		\]
		Thus \(V\) is an \((\alpha,a,\epsilon)\)-block-encoding of
		both matrices. Moreover,
		\(\omega_B=\omega_C=\omega\), while \(B\) is Hurwitz and
		every eigenvalue of \(C\) has positive real part.
		The same block-encoding and parameters are therefore
		compatible with both Hurwitz and non-Hurwitz matrices
		having the same positive gap.
	
\section{Quantum algorithm for computing matrix sign}\label{sec:matrix_sign}
	The matrix sign function has many formulae, such as:
	\begin{align}
		\opn{sign}(A) = A(A^2)^{-\frac 1 2} 
		&= \frac{2}{\pi} \int_{0}^{\infty} \frac{A}{A^2+t^2I}\dd{t} \label{eq:sign_intgral_std}\\ 
		&=\frac{2}{\pi} \int_{0}^{\frac \pi 2} \frac{A} {\cos^2(\theta)A^2+\sin^2(\theta)I} \dd{\theta}. \label{eq:sign_intgral_arctan}
	\end{align}
	See \cite[Chapter 5]{higgs09} for more information about $\sign[A]$ and how its computation is done classically.

	We approximate $\sign[A]$ via its principal Pad{\'e} iteration function,
	\begin{align}
		g_r(x)&= \frac{(1+x)^r-(1-x)^r}{(1+x)^r+(1-x)^r} \label{eq:pade_alg} \\
		&= \begin{cases}
			\frac 2 {2k} \sum_{j=0}^{k-1} \frac x {\sin^2\round*{\frac{2j+1}{2k}\frac{\pi}{2}}+\cos^2\round*{\frac{2j+1}{2k}\frac{\pi}{2}} x^2 } & r=2k \\
				\frac 2 {2k+1} \sum_{j=0}^{k-1} \frac x {\sin^2\round*{\frac{2j+1}{2k+1}\frac{\pi}{2}}+\cos^2\round*{\frac{2j+1}{2k+1}\frac{\pi}{2}} x^2 } + \frac x {2k+1} & r=2k+1
		\end{cases} \label{eq:pade_fraction}
	\end{align}	
	where $r\in\N$, see \cite[Theorem~5.9]{higgs09} and \cite{Kenny91}. 
	One can see that \cref{eq:pade_fraction} looks like a discretization of \cref{eq:sign_intgral_arctan}. 
	We are interested in $\sign^+(A)=0.5(I+\sign[A])$. 
	Define 
	\begin{equation}\label{eq:pade_alg+} 
		g_r^+(x)=\frac 1 2(1+g_r(x))=\frac{(1+x)^r}{(1+x)^r+(1-x)^r}
	\end{equation}
	as an approximation of $\sign^+(A)$. 
	The roots of the denominator,
	\(
	p_r(z)=(1+z)^r+(1-z)^r
	\)
	lie on the imaginary axis. Indeed, if
	\(p_r(z)=0\), then
	\(
	|1+z|=|1-z|,
	\)
	which implies \(\Re(z)=0\).  It follows that, if \(A\) has no purely
	imaginary eigenvalues, then $g_r(A),g^+_r(A)$ are well defined. 
	
	For normal matrices, scalar approximation bounds give matrix norm error bounds directly.
	We next establish bounds that also apply to non-normal matrices.
	
	\begin{theorem}\label{thm:pade_appox}
		Suppose that $\norm{A}\leq 1$ and no eigenvalue of $A$ is purely imaginary. Let $\omega= \min\set{\sigma_{\min}\round{A-itI}: t\in \R}$. 
		Let  $r\in \N$. Let $\eta \in (0,\omega)$ and $R>1$. 
		
		The following hold:
		\begin{enumerate}
			\item For $z\in \C$ with $|z|\leq R$ and $|\Re\round{z}|\geq  \eta$, let $s_{\eta,R}= \round*{\frac{1-2\eta+R^2}{1+2\eta+R^2}}^{0.5}$.
            Note $s_{\eta,R}<1$. 
			Then,
			\[
				\abs{\sign(z)- g_r(z)} \leq \frac{2s_{\eta,R}^r}{1-s_{\eta,R}^r}
			\]
			\item 
			We have,
			\[
				\norm{\sign(A)- g_r(A)} \leq \round*{\frac{2\sqrt{R^2-\eta^2}}{\pi(\omega -\eta)} +\frac{R}{R-1}} \frac{2s_{\eta,R}^r}{1-s_{\eta,R}^r}
			\]
			\item Let $\epsilon \in (0,1]$. If we pick $\eta=\frac{\omega}{2}$, $R=1+\frac{\omega}{2}$ then, we can pick
			\[
				r= O\round*{\frac{1}{\omega}\log\round*{\frac{1}{\omega\epsilon}} }
			\]
			so that $\norm{\sign(A)- g_r(A)}  \leq \epsilon$.
		\end{enumerate} 
	\end{theorem}
	\begin{proof}
		Item 1: Since $\sign(z)$ and $g_r(z)$ are odd we need only show the result for $\Re\round{z}\geq  \eta$. In this case $\sign(z)=1$. Let $q(z)=\frac{1-z}{1+z}$, we can see that
		\[
			\abs{q(z)}^2=\frac{1-2\Re(z) +|z|^2}{1+2\Re(z) +|z|^2} \leq \frac{1-2\eta+R^2}{1+2\eta+R^2} =s_{\eta,R}^2,
		\]
		since the second expression is decreasing in $\Re(z)$ and increasing in $\abs{z}$.

		One can show that
		\[
			\abs{1-g_r(z)} = \abs*{\frac{2q(z)^r}{1+q(z)^r}} \leq \frac{2\abs{q(z)}^r}{1-\abs{q(z)}^r} \leq \frac{2s_{\eta,R}^r}{1-s_{\eta,R}^r}.
		\]

		Item 2: Let $\Gamma_+= \Gamma_v \cup \Gamma_R$ and $\Gamma_{-}={-\Gamma_+}$ be positively oriented where $\Gamma_v=\set{z: |z|\leq R, \Re(z)=\eta }$, \[\Gamma_R = \set{z: |z|= R, \Re(z)\geq \eta }\] and for a set $B\subseteq \C$ $-B=\set{-z: z\in B}$. 
		Since $\norm{A}\leq 1$, $\eta \in (0,\omega)$ and $R>1$, we know that all eigenvalues of $A$ lie in the interior of $\Gamma=\Gamma_{+}\cup \Gamma_{-}$.

		It follows that
		\[
			\sign[A]- g_r(A)= -\frac{1}{2\pi i} \round*{\int_{\Gamma_+}(g_r(z)-1)(zI-A)^{-1} \dd{z} + \int_{\Gamma_-}(g_r(z)+1)(zI-A)^{-1} \dd{z}}
		\]
		and taking norms and using Item 1 yields,
		\[
			\norm{\sign[A]- g_r(A)} \leq \frac{2s_{\eta,R}^r}{1-s_{\eta,R}^r} \frac{1}{2\pi} \round*{\int_{\Gamma_+}\norm{(zI-A)^{-1}} \abs{\dd{z}} + \int_{\Gamma_-}\norm{(zI-A)^{-1}} \abs{\dd{z}}}.
		\]
		So we need only bound these integrals for the result. 
		Due to symmetry a bound for the $\Gamma_{+}$ integral is also a bound for the $\Gamma_{-}$ one. 
		So we only need to bound the former: 
		\begin{equation}\label{eq:Gamma_+}
			\int_{\Gamma_+}\norm{(zI-A)^{-1}} \abs{\dd{z}} = \int_{\Gamma_v}\norm{(zI-A)^{-1}} \abs{\dd{z}}  +\int_{\Gamma_R}\norm{(zI-A)^{-1}} \abs{\dd{z}}. 
		\end{equation}
		
		To bound the integrands we first note that
		\[
			\sigma_{\min}(B+C) \geq \sigma_{\min}(B)- \norm{C}
		\]
		and 
		\[
			\norm{B^{-1}}= \frac{1}{\sigma_{\min}(B)}.
		\]
		
		On $\Gamma_v$, $z= \eta +ti$, recall that $\omega=\min\set{\sigma_{\min}\round{A-itI}: t\in \R}$ and we can see
		\[
			\sigma_{\min}(zI-A)= \sigma_{\min}(A-itI -\eta I) \geq \omega -\eta.
		\]
		Therefore, 
		\[
			\int_{\Gamma_v}\norm{(zI-A)^{-1}} \abs{\dd{z}}  \leq \int_{\Gamma_v}\frac{1}{\omega -\eta} \abs{\dd{z}} =\frac{2\sqrt{R^2-\eta^2}} {{\omega -\eta} }
		\]
		
		On $\Gamma_R$, $\abs{z}=R>1\geq \norm{A}$, so we can use the Neumann-series bound:
		\[
			\norm{(zI-A)^{-1}} \leq \frac{1}{R-\norm{A}}\leq \frac{1}{R-1}
		\]
		The arc-length of $\Gamma_R$ is $2R\arccos(\frac \eta R )\leq \pi R$, hence
		\[
			\int_{\Gamma_R}\norm{(zI-A)^{-1}} \abs{\dd{z}}  \leq \int_{\Gamma_R}\frac{1}{R-1}\abs{\dd{z}} \leq \frac{\pi R}{R-1}. 
		\]
		
		Therefore, 
		\[
			\int_{\Gamma_+}\norm{(zI-A)^{-1}} \abs{\dd{z}} + \int_{\Gamma_-}\norm{(zI-A)^{-1}} \abs{\dd{z}} \leq 2\round*{ \frac{2\sqrt{R^2-\eta^2}} {{\omega -\eta} } +\frac{\pi R}{R-1}}.
		\]
		And so,
		\[
			\norm{\sign[A]- g_r(A)} \leq \round*{\frac{2\sqrt{R^2-\eta^2}}{\pi(\omega -\eta)} +\frac{R}{R-1}} \frac{2s_{\eta,R}^r}{1-s_{\eta,R}^r}
		\]
		as required. 
		
		Item 3: Note that since $\norm{A}\leq 1 $ we must also have $\omega \leq 1$. 
		
		Let 
		\begin{align*}
			C&=\frac{2\sqrt{R^2-\eta^2}}{\pi(\omega -\eta)} +\frac{R}{R-1} 
			 \leq \frac{2R}{\pi(\omega -\eta)} +\frac{R}{R-1}.
		\end{align*}
		Then, picking $\eta=\frac{\omega}{2}$, $R=1+\frac{\omega}{2}$ and noting $\omega \leq 1$ we get
		\begin{align*}
			C&\leq 2\frac{2+\omega}{\pi\omega} +1+\frac 2 \omega= \frac{4+2\pi + \pi \omega +2\omega}{\pi\omega}\leq \frac{6+3 \pi}{\pi}\frac{1}{\omega}
		\end{align*}
		Let $M_1=\frac{6+3 \pi}{\pi}$ and we can see that $C\leq \frac{M_1}{\omega}$.
		
		One can compute that 
		\[
			s_{\eta,R}^r=\exp\round*{-r\opn{arctanh}\round*{\frac{2\eta}{1+R^2}}}
		\]

		When $\eta=\frac{\omega}{2}$, $R=1+\frac{\omega}{2}$
		\[
		\opn{arctanh}\round*{\frac{2\eta}{1+R^2}} =\opn{arctanh}\round*{2\frac{\round{\frac \omega 2}} {1+(1+\round{\frac \omega 2})^2}}=b(\omega).
		\]
		Note $\operatorname{arctanh}(u)\geq u$ for $u\in[0,1)$. Then,
		\[
			b(\omega)=\opn{arctanh}\round*{\frac{\omega}{1+(1+{\frac \omega 2})^2}}
			\geq \frac{\omega}{1+(1+{\frac \omega 2})^2} 
			\geq \frac{\omega}{1+(1+{\frac 1 2})^2}=\frac 4 {13} \omega. 
		\]
		Let $M_2=\frac{4}{13}$.
		
		It follows that the error bound has
		\[
		\round*{\frac{2\sqrt{R^2-\eta^2}}{\pi(\omega -\eta)} +\frac{R}{R-1}} \frac{2s_{\eta,R}^r}{1-s_{\eta,R}^r} 
		\leq  \frac{M_1}{\omega} \frac{2}{\exp\round*{rM_2\omega}-1},
		\]
		setting the RHS $\leq \epsilon$ and rearranging yields 
		\begin{equation}\label{eq:r_complexity}
			r \geq \frac{1}{M_2\omega}\log\round*{1+\frac{2M_1}{\omega\epsilon}}.
		\end{equation}
		For $\epsilon \leq 1$, the lower bound is $O\round*{\frac{1}{\omega}\log\round*{\frac{1}{\omega\epsilon}} }$. 
		
	\end{proof}
	
	Note that
	\[
		\norm{\sign[A]- g_r(A)} = 2\norm{\sign^+(A)- g_r^+(A)}. 
	\]
	So the lower bound on $r$, for $\epsilon$ accuracy, in item 3 of \cref{thm:pade_appox}  applies to approximating $\sign^+(A)$ with $g_r^+(A)$.

	In the remainder of this section we approximate $g_r^+(A)$, and thus $\sign^+(A)$, through a linearization that avoids a linear combination of unitaries over quadrature points; to our knowledge, this particular linearization has not appeared in the quantum literature. 
	Other works, such as \cite{wang2026signembeddingquantumalgorithms}, work by approximating an integral form such as \cref{eq:sign_intgral_std}. Alternatively, in \cite{rodenasruiz2026quantumsolversnonlinearmatrix} they approximate $\sign^+(A)$ by approximating Cauchy's integral formula. 
	The typical approach to approximating these kinds of integrals for this (see for example \cite{TakahiraOhashiSogabeUsuda2022ContourBlockEncoding}) is:
	\begin{enumerate*}
		\item Invert one or more matrices, to form the matrix inverses in the integrand. 
		\item Do LCU, to get a sum which approximates the integral. 
	\end{enumerate*}
	If the second step is done with little care, the scaling parameter of the resulting block-encoding might depend poorly on the number of quadrature points, and become unacceptably large. 
	
	We avoid the second step by choosing the matrix to be inverted with care. 
	This idea has precedent in the classical literature, see \cite{frommer2021matrixfunctionslinearsystems,perez2022linearizations}. 
	Given a function $f:\C \to \C$ that is holomorphic away from its singularities, one can often find a continued fraction representation of $f$. 
	Truncating the continued fraction gives a rational approximation, say $p_N(z)$, to $f$. 
	In turn $p_N(A)$ can be read off the inverse of
	\[
		L=\sum_{j=0}^{\ell} T^{(j)} \otimes A^j
	\]
	where the $T^{(j)}$ are tridiagonal and found from the continued fraction. Again see \cite{frommer2021matrixfunctionslinearsystems} for the details. 
	Assuming the condition number of $L$ and $\ell$ are small, such an approach seems amenable to block-encoding. 
	
	Although $g_r$ arises from the continued fraction of $\sign$, our construction of $L$ does not; it rests on a geometric-series identity.

	We were not able to eliminate the linear combination of unitaries entirely; however, outside the quantum singular value transformation (QSVT) it is used only $O(1)$ times. 
	 
	In the following proposition, we define a useful and well conditioned $L$.

	\begin{proposition}\label{thm:L_matrix}
		Let $\ell\in \N$ and $r=2^\ell$, define $S_r\in \R^{r \times r}$ by 
		\begin{equation}\label{eq:shift_unitary}
			S_r= \sum_{j=0}^{r-2} \ket{j+1}\bra{j} -\ket{0}\bra{r-1}.
		\end{equation}
		Let $A\in \C^{n\times n}$ have no purely imaginary eigenvalues and define $A_{-}=0.5(I-A)$ and $A_+=0.5(I+A)$. 
		
		Define
		\[
			L_r=I_r\otimes A_+ + S_r\otimes A_-.
		\]
		
		The following hold:
		\begin{enumerate}
			\item 
			\[
				L_r^{-1}= \sum_{k=0}^{r-1} (-1)^kS_r^k \otimes A_{+}^{r-1-k}A_{-}^k(A_+^r + A_-^r)^{-1}.
			\]
			\item Consequently,
			\[
				g^+_r(A)= \round*{\bra{0^\ell}\otimes I_n} L_r^{-1} \round*{\ket{0^\ell}\otimes A_+}.
			\]
			\item $\norm{L_r}\leq \sqrt{1+\norm{A}^2}$ and $\sigma_{\min}(L_r)\geq\frac{\omega_A}{\sqrt{1+(\omega_{A}+\norm{A})^2}}$.
		\end{enumerate}
		
	\end{proposition}
	\begin{proof}
		Item 1 holds from direct multiplication, after one notes that $S^r_r=-I$ and that $A_+,A_-$ commute.   
		
		For item 2, one sees,
		\begin{align*}
			\round*{\bra{0^\ell}\otimes I_n} L_r^{-1} \round*{\ket{0^\ell}\otimes I_n}
			&= \round*{\bra{0^\ell}\otimes I_n} \sum_{k=0}^{r-1} (-1)^k\ket{k} \otimes A_{+}^{r-1-k}A_{-}^k(A_+^r + A_-^r)^{-1} \\
			&=A_+^{r-1}(A_+^r + A_-^r)^{-1}. 
		\end{align*}
		Again noting that $A_+,A_-$ and $(A_+^r + A_-^r)^{-1}$ commute; it follows that 
		\[
			\round*{\bra{0^\ell}\otimes I_n} L_r^{-1} \round*{\ket{0^\ell}\otimes A_+}=A_+^{r}(A_+^r + A_-^r)^{-1}=g^+_r(A).
		\]
		
		Item 3: 
		The matrix \(S_r\) is unitary. Its eigenvalues and corresponding
		normalized orthogonal eigenvectors are
		\[
		\theta_k = \frac{(2k+1)}{r}\frac{\pi}{2} 
		\qquad
		\lambda_k=e^{i2\theta_k},
		\qquad
		v_k
		=
		\frac1{\sqrt r}
		\begin{pmatrix}
			1&\lambda_k^{-1}&\cdots&\lambda_k^{-(r-1)}
		\end{pmatrix}^{\mathsf T},
		\]
		for \(k=0,\ldots,r-1\). 
		One can also tell that
		\[
			2L_r=\round*{I_r+S_r}\otimes I_n + \round*{I_r-S_r}\otimes A.
		\]
		It follows that $2L_r$ is unitarily equivalent to a block diagonal matrix with blocks
		\[
			 (1+\lambda_k)I_n + (1-\lambda_k)A
		\]
		for $k=0,\dots, r-1$.
		Thus, we see that
\begin{align*}
 \norm{L_r}=\sigma_{\max}(L_r)
 &=0.5\max\set{\sigma_{\max}((1+\lambda_k)I_n+(1-\lambda_k)A):k=0,\dots,r-1}
\end{align*}
and
\[
 \sigma_{\min}(L_r)
 =0.5\min\set{\sigma_{\min}((1+\lambda_k)I_n+(1-\lambda_k)A):k=0,\dots,r-1}.
\] 
		
		Note that
		\[
			1+e^{i2\theta}= 2e^{i\theta}\cos(\theta)
			\qquad
			1-e^{i2\theta}= -2ie^{i\theta}\sin(\theta)
			\qquad
			\frac{1+e^{i2\theta}}{1-e^{i2\theta}}=i\cot(\theta).
		\]
		
		We now bound $\norm{L_r}$. 
		Define $f(\theta)={\abs{\cos(\theta)} + \norm{A}\abs{\sin(\theta)}}$ and observe that
		\[
			 \norm{L_r}\leq 0.5\max\set{\abs{1+e^{i2\theta}}+\norm{A}\abs{1-e^{i2\theta}}:\theta \in [0,\pi) }=\max\set{f(\theta):\theta \in [0,\pi) }.
		\]
		We can see that 
		\[
		f(\theta)= (\abs{\cos{\theta}},\abs{\sin{\theta}} )\cdot (1,\norm{A})\leq \norm{(\abs{\cos\theta},\abs{\sin\theta})}\,\norm{(1,\norm{A})}=\sqrt{\norm{A}^2+1}; 
		\]
		which is the claimed bound.

		To bound $\sigma_{\min}(L_r)$, first note that $\theta_k\in [\frac{\pi}{2r},\pi-\frac{\pi}{2r}]:=\mathcal I_r$ for all $k=0,\dots, r-1$; so $\sin(\theta_k)\neq 0$.  
		It follows that
		\begin{align*}
			\sigma_{\min}(L_r)&\geq 0.5\min\set{\sigma_{\min}((1+e^{i2\theta})I+(1-e^{i2\theta})A): \theta \in \mathcal I_r } \\
			&= \min\set{\abs{\sin(\theta)}\sigma_{\min}(A+i\cot(\theta)I): \theta \in \mathcal I_r }.
		\end{align*}
		Now let $t=\cot(\theta)$, so $\abs{\sin(\theta)}=\frac{1}{\sqrt{1+t^2}}$. 
		We seek to bound 
		\[
			f_2(t)=\sigma_{\min}(A+itI)\frac{1}{\sqrt{1+t^2}}
		\]
		from below, for $t\in\R $. 
		When $\abs{t}\leq \omega_{A} +\norm{A}$ we have 
		\[
			f_2(t)\geq \omega_A\frac{1}{\sqrt{1+t^2}} 
			\geq \frac{\omega_A}{\sqrt{1+(\omega_{A}+\norm{A})^2}}.
		\]
		
		Note the function $h(x)=\frac{x-\norm{A}}{\sqrt{1+x^2}}$ is increasing for $x\geq 0$. 
		Thus, when $\abs{t}> \omega_{A}+\norm{A}$ we have 
		\[
			f_2(t)
			\geq \frac{\round{\abs{t}-\norm{A}}}{\sqrt{1+t^2}} 
			> \frac{\omega_{A}+\norm{A}-\norm{A}}{\sqrt{1+(\omega_{A}+\norm{A})^2}}
			=\frac{\omega_A}{\sqrt{1+(\omega_{A}+\norm{A})^2}}.
		\]
		In all cases we have the claimed bound.
		
	\end{proof}
	
	Now that we have established that $L_r$ is useful, we can explore how to form a block-encoding of it. 

	\begin{lemma}[Block-encoding of \(L_r\)]\label{thm:form_L_quantum}
		Let \(U_A\) be a \((1,a,0)\)-block-encoding of
		\(A\in\C^{n\times n}\), and let \(r=2^\ell\), where
		\(\ell\in\N\). Define
		\[
		A_\pm=\frac{I_n\pm A}{2},
		\qquad
		L_r=I_r\otimes A_+ + S_r\otimes A_-.
		\]
		Then,
		\begin{enumerate}
			\item Define
			\[
			D_S
			=
			\ket0\bra0\otimes I_r
			+
			\ket1\bra1\otimes S_r.
			\]
			The unitaries \(D_S\) and \(D_S^\dagger\) can each be
			implemented exactly using \(O(\ell)\) elementary gates and
			one reusable work qubit.
			
			\item One can construct a
			\((\sqrt2,a+1,0)\)-block-encoding \(U_L\) of \(L_r\); its adjoint \(U_L^\dagger\) block-encodes \(L_r^\dagger\).
			Each application of \(U_L\) ( \(U_L^\dagger\)) uses one application of \(D_S\) (\(D_S^\dagger\)), one controlled application of \(U_A\) (\(U_A^\dagger\)), and three Hadamard gates. Note $\norm{A}\leq1$, so $\norm{L_r}\leq \sqrt{2}$, by \cref{thm:L_matrix}.
		\end{enumerate}
	\end{lemma}
	\begin{proof}
		For item \(1\), define the modular incrementer
		\[
		C_r
		=
		\sum_{j=0}^{r-2}\ket{j+1}\bra j
		+\ket0\bra{r-1}
		\]
		and the reflection
		\[
		R_r=I_r-2\ket{r-1}\bra{r-1}.
		\]
		It follows that,
		\[
		S_r=C_rR_r.
		\]
		
		Let
		\[
		D_C
		=
		\ket0\bra0\otimes I_r
		+
		\ket1\bra1\otimes C_r,
		\qquad
		D_R
		=
		\ket0\bra0\otimes I_r
		+
		\ket1\bra1\otimes R_r.
		\]
		It follows that
		\[
		D_S=D_CD_R.
		\]
		The operator \(D_C\) is the singly controlled incrementer on an
		\(\ell\)-qubit register. By
		\cite[Corollary~7]{Vandaele26}, it has an exact
		\(O(\ell)\)-gate implementation using one work qubit whose
		initial state may be arbitrary and which is restored at the end.
		
		Since \(r=2^\ell\), we have
		\[
		R_r=I_r-2\ket{1^\ell}\bra{1^\ell},
		\]
		and hence
		\[
		D_R
		=
		I_{2r}
		-
		2\ket{1}\bra{1}\otimes
		\ket{1^\ell}\bra{1^\ell}.
		\]
		Thus \(D_R\) is the phase gate conditioned on all \(\ell+1\)
		qubits being in the state \(\ket1\). Conjugating one of these
		qubits by Hadamard gates reduces its implementation to an
		\(\ell\)-controlled \(X\) gate. By
		\cite[Corollary~3]{nie2024quantumcircuitmultiqubittoffoli}, this also has an exact
		\(O(\ell)\)-gate implementation using one reusable work qubit.
		The same work qubit can be used successively for \(D_R\) and
		\(D_C\). Therefore \(D_S=D_CD_R\) uses \(O(\ell)\) gates.
		Reversing this circuit gives an implementation of
		\(D_S^\dagger\) with the same cost.
		
		For item \(2\), let \(c\) be a selector qubit, let \(J\) denote
		the \(\ell\)-qubit index register, let \(B\) denote the
		\(a\) block-encoding ancillas of \(U_A\), and let \(X\) denote
		the system register on which \(A\) acts. Define
		\[
		H=\frac{1}{\sqrt2}
		\begin{pmatrix}
			1&1\\
			1&-1
		\end{pmatrix},
		\qquad
		H_c=H\otimes I_{JBX},
		\]
		so that \(H_c\) applies a Hadamard gate to \(c\) and acts as the
		identity on every other register. Also define
		\[
		C_A
		=
		\round*{\ket0\bra0}_c\otimes I_{JBX}
		+
		\round*{\ket1\bra1}_c\otimes(I_r\otimes U_A).
		\]
		This is the controlled application of \(U_A\), with the index
		register left unchanged. When used on all four registers,
		\(D_S\) is implicitly tensored with \(I_{BX}\).
		
		Set
		\[
		U_L=H_cD_SH_cC_AH_c.
		\]
		A high level circuit diagram of $U_L$ is given in \cref{fig:block-encoding-Lr}.
		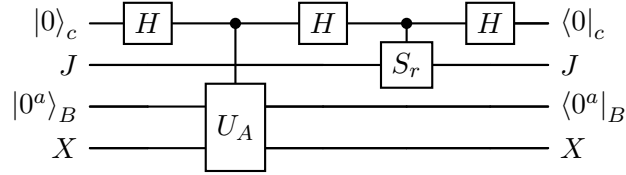
\begin{figure}[ht]
			\centering
			\begin{quantikz}[
				row sep={0.55cm,between origins},
				column sep=0.45cm
				]
				\lstick{$\ket{0}_c$}
				& \gate{H}
				& \ctrl{2}
				& \gate{H}
				& \ctrl{1}
				& \gate{H}
				& \rstick{$\bra{0}_c$} \qw \\
				\lstick{$J$}
				& \qw
				& \qw
				& \qw
				& \gate{S_r}
				& \qw
				& \rstick{$J$} \qw \\
				\lstick{$\ket{0^a}_B$}
				& \qw
				& \gate[wires=2]{U_A}
				& \qw
				& \qw
				& \qw
				& \rstick{$\bra{0^a}_B$} \qw \\
				\lstick{$X$}
				& \qw
				& \ghost{U_A}
				& \qw
				& \qw
				& \qw
				& \rstick{$X$} \qw
			\end{quantikz}
			\caption{
				Circuit for
				\(U_L=H_cD_SH_cC_AH_c\).
				Here \(c\) is the selector qubit, \(J\) is the
				\(\ell\)-qubit index register, \(B\) contains the \(a\)
				block-encoding ancillas of \(U_A\), and \(X\) is the system
				register. Taking the
				\(\ket{0}_c\ket{0^a}_B\) matrix element gives
				\(L_r/\sqrt{2}\).
			}
			\label{fig:block-encoding-Lr}
		\end{figure}
		
		Relative to the computational basis of \(c\), the three
		relevant unitaries have the operator-block representations
		\[
		H_c
		=
		\frac1{\sqrt2}
		\begin{pmatrix}
			I&I\\
			I&-I
		\end{pmatrix},
		\qquad
		D_S
		=
		\begin{pmatrix}
			I&0\\
			0&S_r\otimes I_{BX}
		\end{pmatrix},
		\]
		and
		\[
		C_A
		=
		\begin{pmatrix}
			I&0\\
			0&I_r\otimes U_A
		\end{pmatrix}.
		\]
		Multiplying these blocks gives
		\begin{align*}
			\bra0_cU_L\ket0_c
			&=
			\frac{1}{2\sqrt2}
			\left[
			(I_r+S_r)\otimes I_{BX}
			+
			(I_r-S_r)\otimes U_A
			\right].
		\end{align*}
		Because \(U_A\) is a \((1,a,0)\)-block-encoding of \(A\),
		\[
		(\bra{0^a}_B\otimes I_X)
		U_A
		(\ket{0^a}_B\otimes I_X)
		=A.
		\]
		Consequently,
		\begin{align*}
			&(\bra0_c\bra{0^a}_B\otimes I_{JX})
			U_L
			(\ket0_c\ket{0^a}_B\otimes I_{JX})\\
			&\qquad=
			\frac{1}{2\sqrt2}
			\left[
			(I_r+S_r)\otimes I_n
			+
			(I_r-S_r)\otimes A
			\right]\\
			&\qquad=
			\frac{1}{\sqrt2}
			\left[
			I_r\otimes\frac{I_n+A}{2}
			+
			S_r\otimes\frac{I_n-A}{2}
			\right]
			=
			\frac{L_r}{\sqrt2}.
		\end{align*}
		Hence \(U_L\) is a
		\((\sqrt2,a+1,0)\)-block-encoding of \(L_r\).
		
		The displayed construction contains one \(D_S\), one controlled
		application of \(U_A\), and three Hadamard gates. Reversing the
		circuit gives the corresponding resource count for
		\(U_L^\dagger\).
	\end{proof}
	
\subsection{QSVT inversion}

	\begin{proposition}[QSVT inversion]
		\label{thm:QSVT_inv}
		Let \(U_B\) be an exact \((\beta,b,0)\)-block-encoding of an
		invertible matrix \(B\). Suppose that a number \(\eta>0\) is known
		such that
		\[
		\eta\leq \sigma_{\min}(B)
		\qquad\text{and}\qquad
		\eta\leq \frac{\beta}{2}.
		\]
		Then, for every
		\[
		0<\epsilon\leq \min\{1,2\beta^{-1}\},
		\]
		one can construct a
		\[
		\left(\frac{2}{\eta},\,b+1,\,\epsilon\right)
		\]
		block-encoding of \(B^{-1}\). The construction uses
		\[
		O\left(
		\frac{\beta}{\eta}
		\log\left(\frac{2}{\eta\epsilon}\right)
		\right)
		\]
		total queries to \(U_B\), \(U_B^\dagger\).
		
		Let $d$ denote this query count. In addition, the QSVT circuit uses \(O(d)\) one-qubit phase gates and \(O(d)\) projector-controlled NOT gates. For a standard block-encoding, these controlled gates can be decomposed into
		$$ O\bigl(d(b+1)\bigr) $$
		bounded-arity elementary gates using one reusable dirty work qubit.
	\end{proposition}
	
	\begin{proof}
		Set
		\[
		T:=\frac{B}{\beta},
		\qquad
		\delta:=\frac{\eta}{\beta}.
		\]
		Thus \(U_B\) is an exact \((1,b,0)\)-block-encoding of \(T\), and
		\[
		\sigma_{\min}(T)\geq\delta,
		\qquad
		0<\delta\leq\frac12.
		\]
		Define
		\[
		\epsilon_{\mathrm{QSVT}}
		:=\frac{\eta\epsilon}{2}.
		\]
		Since \(\epsilon\leq2\beta^{-1}\),
		\[
		\epsilon_{\mathrm{QSVT}}
		=\frac{\eta\epsilon}{2}
		\leq\frac{\eta}{\beta}
		=\delta.
		\]
		
		By the singular-value inversion result of
		\cite[Theorem~41]{QSVT}, QSVT applied to \(T\) produces, using one
		additional signal qubit, a unitary whose encoded block \(W\)
		satisfies
		\[
		\left\|
		W-\frac{\delta}{2}T^{-1}
		\right\|
		\leq\epsilon_{\mathrm{QSVT}}.
		\]
		Since \(T^{-1}=\beta B^{-1}\),
		\[
		\frac{\delta}{2}T^{-1}
		=
		\frac{\eta}{2\beta}\,\beta B^{-1}
		=
		\frac{\eta}{2}B^{-1}.
		\]
		Consequently,
		\[
		\left\|
		B^{-1}-\frac{2}{\eta}W
		\right\|
		\leq
		\frac{2}{\eta}\epsilon_{\mathrm{QSVT}}
		=\epsilon.
		\]
		Conjugating the additional signal qubit by Hadamard gates
		converts the \(\ket{+}\)-projected block in
		\cite[Theorem~41]{QSVT} into a standard
		\(\ket{0}\)-projected block.
		Hence the resulting unitary is a
		\((2/\eta,b+1,\epsilon)\)-block-encoding of \(B^{-1}\).
		
		The polynomial degree, and therefore the number of applications of
		\(U_B\) and \(U_B^\dagger\), is
		\[
		O\left(
		\delta^{-1}
		\log\left(\epsilon_{\mathrm{QSVT}}^{-1}\right)
		\right)
		=
		O\left(
		\frac{\beta}{\eta}
		\log\left(\frac{2}{\eta\epsilon}\right)
		\right).
		\]
		Let \(d\) denote the above query count. In addition,
		for the gate count, \cite[Theorem~41]{QSVT} uses \(O(d)\) controlled operations associated with the input and output projectors. These are \(b\)-controlled NOTs, up to negative-control \(X\) gates. They cost \(O(b)\) elementary gates each using one dirty ancilla by \cite[Corollary~3]{nie2024quantumcircuitmultiqubittoffoli}. 
		
	\end{proof}

\subsection{Projector construction complexity}
	Combining the previous two results gives the following complexity
	bounds for block-encoding an approximation to $\sign^+(A)$.
	\begin{theorem}\label{thm:sign_complexity}
		Suppose that $\norm{A}\leq 1$ and no eigenvalue of $A$ is purely imaginary.
		Let $\omega= \min\set{\sigma_{\min}\round{A-itI}: t\in \R}$, $\nu \in (0,\omega]$, and $U_A$ be a $(1,a,0)$-block-encoding of $A$. 
		Let $r=2^\ell$, $\ell\in \N$.
		The following hold:
		\begin{enumerate}
			\item For all $\epsilon_{inv} \in (0,1]$, we can construct a $\round*{\frac{2\sqrt{5}}{\nu},2a+\ell+3,\epsilon_{inv}}$-block-encoding of $g^+_{r}(A)$ which uses $d=O\round*{\frac{1}{\nu}\log\round*{\frac 1 {\nu\epsilon_{inv}}}}$ queries of $U_A$ and $U_A^{\dagger}$, and uses an additional $O(d(\ell+a))$ elementary gates.
			
			\item For all $\epsilon_{+} \in (0,1]$, we can pick
                \[r=2^\ell=O\round*{ \frac{1}{\nu}\log\round*{\frac 1 {\nu\epsilon_{+}}} }\] 
                to construct a  $\round*{\frac{2\sqrt{5}}{\nu},2a+\ell+3,\epsilon_{+}}$-block-encoding of $\sign^+(A)$ which uses
                \[d=O\round*{\frac{1}{\nu}\log\round*{\frac {1} {\nu\epsilon_{+}}}}\]
                queries of $U_A$ and $U_A^{\dagger}$, and uses an additional $O(d(\log(d)+a))$ elementary gates.

		\end{enumerate}
	\end{theorem}
	\begin{proof}
		Item 1: \Cref{thm:form_L_quantum} gives us unitary $U_L$ which is a $(\sqrt{2},a+1,0)$-block-encoding of $L_r$. 
		Using $\nu \in (0,\omega]$, recalling that $\omega,\norm{A}\leq 1$ and item 3 of \cref{thm:L_matrix} we get 
		\[
			\sigma_{\min}(L_r)\geq  \frac{\omega}{\sqrt{1+(\omega+\norm{A})^2}} \geq \frac{\omega}{\sqrt{1+(1+1)^2}}=\frac{\omega}{\sqrt{5}}\geq \frac{\nu}{\sqrt{5}}.
		\] 
		 
		Using this lower bound, we can apply \cref{thm:QSVT_inv}  to $L_r$.
		Here \(\eta=\nu/\sqrt5\), \(\beta=\sqrt2\), and $\eta \leq 5^{-0.5}\leq \beta/2$. 
		Since
		\(\epsilon_{inv}\leq1\leq2/\beta\), \cref{thm:QSVT_inv}
		gives a
		$(2\sqrt 5 /\nu, a+2,\epsilon_{inv})$-block-encoding of $L_r^{-1}$. 
		This costs $d=O\round*{\frac{1}{\nu}\log\round*{\frac 1 {\nu\epsilon_{inv}}}}$ applications of $U_L,U_L^{\dagger}$ and $O(d(a+1))$ elementary gates. 
		By \cref{thm:form_L_quantum}, applications of $U_L$ ($U_L^\dagger$) cost 1 application of $U_A$ ($U_A^\dagger$) for a total of $d$ calls to $U_A,U_A^{\dagger}$. 
		By the same result, applying $U_L$ or $U_L^\dagger$ also costs $O(\ell)$ elementary gates, so the total elementary gate cost is $O(d(\ell+a))$.
		
		To get a block encoding of $g_r^+(A)$ we seek to apply item 2 of \cref{thm:L_matrix}:
		\[
			g^+_r(A)= \round*{\bra{0^\ell}\otimes I_n} L_r^{-1} \round*{\ket{0^\ell}\otimes A_+},
		\] 
		where $A_+=0.5(I+A)$. We therefore, want to multiply $L_r^{-1}$ by $C=I_r\otimes A_+$. 
		Since
		\[
		C=\frac12(I_r\otimes I_n)+\frac12(I_r\otimes A),
		\]
		\cite[Lemma~52]{QSVT}, applied with coefficient vector
		\((1/2,1/2)\), gives an exact \((1,a+1,0)\)-block-encoding
		of \(C\). It uses one controlled application of \(U_A\) and two
		Hadamard gates.
		Forming the product $L_r^{-1}C$, results in a $((1)2\sqrt 5/\nu, a+2+a+1,\epsilon_{inv})$-block-encoding thereof, \cite[Lemma~53]{QSVT}. 
		This does not use any additional gate resources. 
		
		Finally, regarding the \(\ell\)-qubit index register as an
		additional block-encoding register and projecting it onto
		\(\ket{0^\ell}\) gives a $(2\sqrt 5/\nu, 2a+3+\ell,\epsilon_{inv})$-block-encoding of $g^+_r(A)$.

		Item 2 follows from item 1 of this theorem with $\epsilon_{inv}=0.5\epsilon_{+}$ and item 3 of \cref{thm:pade_appox} with $\epsilon=\epsilon_{+}$. 
		Since \(\nu\leq\omega\), the sufficient lower bound in
		\eqref{eq:r_complexity} with \(\omega\) replaced by \(\nu\)
		is no smaller. Choosing \(r\) to be the least power of two above
		this bound changes it by at most a factor of two.
		Thus we can pick $r=2^\ell=O\round*{\frac 1 \nu\log\round*{ {\frac{1}{\nu\epsilon_{+}}}}}$ so that
		\[
			\epsilon_{+} \geq \norm{\sign[A]-g_r(A)} =2\norm{\sign^+(A)-g^+_r(A)}
		\] 
		Hence, if $V$ is the $(\beta,b,\epsilon_{inv})$-block-encoding of $g^+_r(A)$ and $\tilde{g}=\beta\round*{\langle {0^b} \rvert\otimes I}V\round*{\lvert {0^b} \rangle\otimes I}$, we have
		\begin{align*}
			\norm*{\sign^+(A)-\tilde g}
			&\leq
			\norm*{\sign^+(A)-g_r^+(A)}
			+\norm*{g_r^+(A)-\tilde g} \\
			&\leq \frac{\epsilon_{+}}2
			+\epsilon_{inv}
			=\epsilon_{+}.
		\end{align*}
		
		The gate costs remain the same as in item 1, but here $\ell = O(\log(d))$, since we have chosen $r=O(d)$. 
		
	\end{proof}
	The parameter $r$ is a direct proxy to the number of quadrature points of the integral \cref{eq:sign_intgral_arctan} (via \cref{eq:pade_fraction}). 
	We see that our approximate block-encoding scaling parameter is independent of $r$. 
	One may expect such dependence with naive LCU approaches to computing the sum \cref{eq:pade_fraction}. 
	Practically this means higher accuracy requires no trade off in measurement probabilities. 
	For fixed inversion accuracy, the query bound in item 1 of \cref{thm:sign_complexity} is independent of $r$; the additional gate bound depends on $r$ through $\ell=\log_2 r$. 
	Higher target accuracy also increases the query bound through the required inversion accuracy. 

\section{Proof of \cref{thm:search-hybrid}}
	\label{sec:Proof_thm:search-hybrid}
    This proof is a step-by-step form of the hybrid argument of
	Bennett et al.\ \cite[proof of Theorem~3.3]{BennettEtAl1997},
	retaining the size $\eta$ of the local oracle difference.
	The usual phase-oracle case is discussed in
	\cite[Exercise~11.12]{deWolf2026}. For a general framework
	of lower bounds for unitary oracles, including inverse and
	controlled access, see
	\cite[Section~2 and Theorem~37]{Belovs2015}.
	
		Let $T$ be the worst-case number of queries. We prove the
		lower bound even when one query may coherently choose
		between doing nothing, applying $U_x$, and applying
		$U_x^\dagger$. Specifically, use the state space
		\[
		\mathcal H
		=\C^3\otimes\C^d\otimes\C^n\otimes\mathcal A,
		\]
		whose factors represent the query selector, the local
		state acted on by $V_0$ or $V_1$, the index register,
		and an arbitrary ancillary space $\mathcal A$, respectively.
		The enlarged query operator on $\mathcal H$ is
		\[
		\widehat U_x
		=
		\left(
		\ket{0}\bra{0}\otimes I_{dn}
		+\ket{1}\bra{1}\otimes U_x
		+\ket{2}\bra{2}\otimes U_x^\dagger
		\right)\otimes I_{\mathcal A}.
		\]
		Here $\ket{0},\ket{1},\ket{2}$ belong to the selector
		register $\C^3$. 
		By preparing the selector register appropriately, one application
		of $\widehat U_x$ can implement the identity, $U_x$, or
		$U_x^\dagger$, including their controlled forms. 
		Therefore, counting each application of 
		$\widehat U_x$ as one query gives at least as much
		oracle access as the model in the statement.
		
		By keeping random choices and intermediate measurement
		outcomes in ancillary registers, all measurements can
		be deferred until the end
		\cite[Exercise~2.8]{deWolf2026}. Branches that finish
		early select the identity block in subsequent queries.
		The computation therefore has the form
		\begin{equation}
			\label{eq:search-hybrid-evolution}
			\ket{\phi_t(x)}
			=\widehat U_x W_t\ket{\phi_{t-1}(x)},
			\qquad t=1,\ldots,T,
		\end{equation}
		where $W_1,\ldots,W_T$ are unitaries on $\mathcal H$
		independent of the unknown input $x$, and
		$\ket{\phi_0(x)}=\ket{\phi_0}$ is the same initial
		unit vector for every $x$. Thus $\ket{\phi_t(x)}$
		denotes the state immediately after query $t$.
		Any unitary after the last query is absorbed into
		the final measurement.
		
		For $j=1,\ldots,n$, define the orthogonal projection
		\[
		\Pi_j
		=I_3\otimes I_d\otimes\ket{j}\bra{j}
		\otimes I_{\mathcal A}.
		\]
		It retains exactly the components whose index register
		equals $\ket{j}$, without restricting any other register.
		Set
		\[
		p_{t,j}
		=\norm{\Pi_j W_t\ket{\phi_{t-1}(\mathbf 0)}}^2.
		\]
		This is the probability of obtaining index $j$ if the
		index register were measured immediately before query
		$t$ in the execution with input $\mathbf 0$; no such
		measurement is performed. Since $\sum_{j=1}^n\Pi_j=I$
		and the projections are mutually orthogonal,
		\[
		\sum_{j=1}^n p_{t,j}=1.
		\]
		
		The inputs $\mathbf 0$ and $e_j$ differ only at
		coordinate $j$, so
		\[
		U_{e_j}-U_{\mathbf 0}
		=(V_1-V_0)\otimes\ket{j}\bra{j}.
		\]
		The same identity holds for their adjoints after
		replacing $V_1-V_0$ by $V_1^\dagger-V_0^\dagger$.
		Since both differences have norm $\eta$, the block
		structure of $\widehat U_x$ gives
		\[
		\widehat U_{e_j}-\widehat U_{\mathbf 0}
		=(\widehat U_{e_j}-\widehat U_{\mathbf 0})\Pi_j,
		\qquad
		\norm{\widehat U_{e_j}-\widehat U_{\mathbf 0}}=\eta.
		\]
		Consequently,
		\begin{equation}
			\label{eq:search-hybrid-one-query}
			\norm{
				(\widehat U_{e_j}-\widehat U_{\mathbf 0})
				W_t\ket{\phi_{t-1}(\mathbf 0)}
			}
			\leq\eta\sqrt{p_{t,j}}.
		\end{equation}
		
		We now compare the two executions with inputs $e_j$
		and $\mathbf 0$, each using its own fixed oracle.
		In \eqref{eq:search-hybrid-evolution}, add and subtract
		$\widehat U_{e_j}W_t\ket{\phi_{t-1}(\mathbf 0)}$.
		The triangle inequality and unitarity give
		\begin{align*}
			\norm{\ket{\phi_t(e_j)}-\ket{\phi_t(\mathbf 0)}}
			&\leq
			\norm{\ket{\phi_{t-1}(e_j)}
				-\ket{\phi_{t-1}(\mathbf 0)}}\\
			&\quad+
			\norm{
				(\widehat U_{e_j}-\widehat U_{\mathbf 0})
				W_t\ket{\phi_{t-1}(\mathbf 0)}
			}\\
			&\leq
			\norm{\ket{\phi_{t-1}(e_j)}
				-\ket{\phi_{t-1}(\mathbf 0)}}
			+\eta\sqrt{p_{t,j}}.
		\end{align*}
		The initial states coincide. Iterating this estimate
		therefore yields
		\[
		\norm{\ket{\phi_T(e_j)}-\ket{\phi_T(\mathbf 0)}}
		\leq\eta\sum_{t=1}^T\sqrt{p_{t,j}}.
		\]
		By Cauchy--Schwarz,
		\begin{align*}
			\sum_{j=1}^n
			\norm{\ket{\phi_T(e_j)}-\ket{\phi_T(\mathbf 0)}}^2
			&\leq
			\eta^2\sum_{j=1}^n
			\left(\sum_{t=1}^T\sqrt{p_{t,j}}\right)^2\\
			&\leq
			\eta^2 T\sum_{t=1}^T\sum_{j=1}^n p_{t,j}
			=\eta^2T^2.
		\end{align*}
		
		Let $M$, with $0\leq M\leq I$, correspond to the final
		output $x=\mathbf 0$. For any unit vectors
		$\ket{u}$ and $\ket{v}$, the inequality
		$\norm{M-\tfrac12 I}\leq\tfrac12$ gives
		\begin{align*}
			\left|\bra{u}M\ket{u}-\bra{v}M\ket{v}\right|
			&=
			\left|\operatorname{Re}\left(
			(\bra{u}-\bra{v})(M-\tfrac12 I)
			(\ket{u}+\ket{v})
			\right)\right|\\
			&\leq
			\tfrac12\norm{\ket{u}-\ket{v}}
			\norm{\ket{u}+\ket{v}}\\
			&\leq\norm{\ket{u}-\ket{v}}.
		\end{align*}
		The success assumption gives
		\[
		\bra{\phi_T(\mathbf 0)}M\ket{\phi_T(\mathbf 0)}
		\geq1-\delta,
		\qquad
		\bra{\phi_T(e_j)}M\ket{\phi_T(e_j)}
		\leq\delta.
		\]
		Applying the preceding inequality to these two
		final states shows that
		\[
		\norm{\ket{\phi_T(e_j)}-\ket{\phi_T(\mathbf 0)}}
		\geq1-2\delta
		\qquad (j=1,\ldots,n).
		\]
		Thus $n(1-2\delta)^2\leq\eta^2T^2$, and therefore
		\begin{equation}
			\label{eq:search-hybrid-explicit}
			T\geq\frac{(1-2\delta)\sqrt n}{\eta}.
		\end{equation}
		This proves the claimed asymptotic bound.
\section{Proof of \cref{thm:rand_diag_matrix}}
\label{sec:Proof_thm:rand_diag_matrix}
        Since the matrices are real diagonal, $\omega_{A}=\min_{j}\abs{a_j}$. 
		Since each $\abs{a_j}$ is independent and uniformly distributed on $[0,1]$, $j=1,\dots n$, the minimum statistic of the sample (which is $\omega_{A}$) is known to have cumulative  distribution function: $\Pr\left(\omega_A\le x\right)=1-(1-x)^n$, where $x\in [0,1]$.  
		So we can see that 
			$\Pr\left(
			l\leq\omega_A\leq u
			\right)
			=\left(1-l\right)^n
			-
			\left(1-u\right)^n$.  Let $y_n=\Pr\left(0.2/n\leq\omega_A\leq2/n\right)$, for $n\geq 2$. 
		Since $(1-x/n)^n$ is increasing in $n$ for fixed $x\in(0,n)$, we have $(1-0.2/n)^n\geq 0.9^2=0.81$ for $n\geq2$ and $(1-2/n)^n\leq e^{-2}<0.136$. Hence $y_n>0.81-0.136>0.67$. 
\section{Proof of \cref{thm:fixed_dimensional_normalization} \label{sec:Proof_thm:fixed_dimensional_normalization}}
	We first prove the upper bound. 
	Use the contour \(\Gamma_+\) from \cref{eq:Gamma_+}, with
	\(\eta=0\) and \(R=2\), as in the proof of
	\cref{thm:pade_appox}. Thus
	\(\Gamma_+=\Gamma_v\cup\Gamma_R\), where \(\Gamma_v\) is the
	imaginary-axis segment joining \(-2i\) to \(2i\), and
	\(\Gamma_R\) is the right semicircle of radius \(2\).
	Since \(\omega_A>0\), the imaginary axis contains no eigenvalue of
	\(A\), and the Riesz projection formula gives
	\[
	\sign^+(A)
	=
	\frac{1}{2\pi i}
	\int_{\Gamma_+}(zI-A)^{-1}\,\dd z.
	\]
	Consequently,
	\begin{equation}
		\label{eq:fixed_dim_contour_bound}
		\norm{\sign^+(A)}
		\leq
		\frac{1}{2\pi}
		\left(
		\int_{\Gamma_v}\norm{(zI-A)^{-1}}\,\abs{\dd z}
		+
		\int_{\Gamma_R}\norm{(zI-A)^{-1}}\,\abs{\dd z}
		\right).
	\end{equation}
	
	On \(\Gamma_R\), the same Neumann-series estimate used in the proof
	of \cref{thm:pade_appox} gives
	\[
	\norm{(zI-A)^{-1}}
	\leq
	\frac{1}{\abs z-\norm A}
	\leq1.
	\]
	Since \(\Gamma_R\) has length \(2\pi\), its contribution to
	\eqref{eq:fixed_dim_contour_bound} is at most \(1\).
	
	It remains to estimate the integral over \(\Gamma_v\). Write the
	eigenvalues of \(A\), with algebraic multiplicity, as
	\[
	\lambda_j=a_j+ib_j,
	\qquad j=1,\ldots,n,
	\]
	and put \(M=\omega_A^{-1}\). By the definition of \(\omega_A\),
	\begin{equation}
		\label{eq:fixed_dim_gap_bound}
		\norm{(itI-A)^{-1}}\leq M
		\qquad(t\in\R).
	\end{equation}
	
	For \(\abs t\leq2\),
	\[
	\norm{itI-A}\leq\abs t+\norm A\leq3.
	\]
	If the singular values of an invertible matrix \(B\) are ordered
	as
	\[
	\sigma_1(B)\geq\cdots\geq\sigma_n(B),
	\]
	then
	\[
	\norm{\opn{adj}(B)}
	=
	\prod_{j=1}^{n-1}\sigma_j(B)
	\leq
	\norm{B}^{n-1}.
	\]
	Applying this to \(B=itI-A\), we obtain
	\begin{align*}
		\norm{(itI-A)^{-1}}
		&=
		\frac{\norm{\opn{adj}(itI-A)}}
		{\abs{\det(itI-A)}}\\
		&\leq
		\frac{3^{n-1}}
		{\displaystyle\prod_{j=1}^n\abs{it-\lambda_j}}\\
		&\leq
		\frac{3^{n-1}}
		{\displaystyle
			\min_{1\leq j\leq n}\abs{t-b_j}^{\,n}}.
	\end{align*}
	Combining this with \eqref{eq:fixed_dim_gap_bound} gives
	\begin{equation}
		\label{eq:fixed_dim_resolvent_sum}
		\norm{(itI-A)^{-1}}
		\leq
		\sum_{j=1}^n
		\min\left\{
		M,\frac{3^{n-1}}{\abs{t-b_j}^{\,n}}
		\right\}.
	\end{equation}
	
	For \(M>0\), splitting the integral at
	\(\abs{t-b}=(3^{n-1}/M)^{1/n}\) gives
	\[
	\int_{\R}
	\min\left\{
	M,\frac{3^{n-1}}{\abs{t-b}^{\,n}}
	\right\}\,\dd t
	=
	\frac{2n}{n-1}
	3^{(n-1)/n}M^{(n-1)/n}.
	\]
	Integrating \eqref{eq:fixed_dim_resolvent_sum} yields
	\[
	\int_{-2}^{2}
	\norm{(itI-A)^{-1}}\,\dd t
	\leq
	\frac{2n^2}{n-1}
	3^{(n-1)/n}
	\omega_A^{-(n-1)/n}.
	\]
	It follows from \eqref{eq:fixed_dim_contour_bound} that
	\[
	\norm{\sign^+(A)}
	\leq
	1+
	\frac{n^2}{\pi(n-1)}
	3^{(n-1)/n}
	\omega_A^{-(n-1)/n}.
	\]
	Since
	\[
	\omega_A
	\leq\sigma_{\min}(A)
	\leq\norm A
	\leq1,
	\]
	we conclude that
	\begin{equation}
		\label{eq:fixed_dim_upper_final}
		\norm{\sign^+(A)}
		\leq
		\left(
		1+
		\frac{n^2}{\pi(n-1)}
		3^{(n-1)/n}
		\right)
		\omega_A^{-(n-1)/n}.
	\end{equation}
	
	We next construct matrices attaining this exponent. For \(s\geq2\),
	put
	\[
	\theta
	=
	\begin{cases}
		0,&n\text{ odd},\\
		\pi/2,&n\text{ even},
	\end{cases}
	\]
	and define
	\[
	A=A_{n,s}
	=
	\begin{pmatrix}
		0&1&&0\\
		&0&\ddots&\\
		&&\ddots&1\\
		e^{i\theta}s^{-n}&0&\cdots&0
	\end{pmatrix},
	\qquad
	D
	=
	\opn{diag}
	\left(1,s^{-1},\ldots,s^{-(n-1)}\right),
	\]
	and
	\[
	S
	=
	\begin{pmatrix}
		0&1&&0\\
		&0&\ddots&\\
		&&\ddots&1\\
		e^{i\theta}&0&\cdots&0
	\end{pmatrix}.
	\]
	The matrix \(S\) is unitary. Its eigenvalues and corresponding
	normalized eigenvectors are
	\[
	\lambda_k
	=
	\exp\left(
	\frac{i(\theta+2\pi k)}{n}
	\right),
	\qquad
	v_k
	=
	\frac1{\sqrt n}
	\begin{pmatrix}
		1&\lambda_k&\cdots&\lambda_k^{n-1}
	\end{pmatrix}^{\mathsf T},
	\]
	for \(k=0,\ldots,n-1\). Indeed,
	\(\lambda_k^n=e^{i\theta}\) and \(Sv_k=\lambda_kv_k\).
	Since \(S\) is unitary and the \(\lambda_k\) are distinct, the
	vectors \(v_k\) form an orthonormal eigenbasis.
	
	Direct multiplication gives
	\begin{equation}
		\label{eq:fixed_dim_similarity}
		D^{-1}AD=s^{-1}S,
		\qquad
		A^n=e^{i\theta}s^{-n}I.
	\end{equation}
	Moreover,
	\[
	A^\dagger A
	=
	\opn{diag}(s^{-2n},1,\ldots,1),
	\]
	so
	\[
	\norm A=1.
	\]
	
	None of the eigenvalues of \(S\) is purely imaginary. If \(n\) is
	odd, then \((\pm i)^n\) is purely imaginary and hence cannot equal
	\(e^{i\theta}=1\). If \(n\) is even, then \((\pm i)^n\) is real and
	hence cannot equal \(e^{i\theta}=i\).
	
	Let
	\[
	K
	=
	\left\{
	k\in\{0,\ldots,n-1\}:
	\Re(\lambda_k)>0
	\right\},
	\]
	and define
	\[
	Q
	=
	\sum_{k\in K}v_kv_k^\dagger.
	\]
	Then \(Q=\sign^+(S)\). By
	\eqref{eq:fixed_dim_similarity}, and because multiplication by the
	positive scalar \(s^{-1}\) does not change the right-half-plane
	spectral subspace,
	\begin{equation}
		\label{eq:fixed_dim_projection_similarity}
		\sign^+(A)
		=
		DQD^{-1}
	\end{equation}
	
	Let $Q=(q_{ij})_{i,j=1}^n$. 
    We now estimate \(q_{1n}\). 
    Since \(\lambda_k^n=e^{i\theta}\),
	\[
	q_{1n}
	=
	\frac1n
	\sum_{k\in K}\overline{\lambda_k^{n-1}}
	=
	\frac{e^{-i\theta}}n
	\sum_{k\in K}\lambda_k.
	\]
	The indices in \(K\) are consecutive modulo \(n\). If
	\(m=\abs K\), the geometric-series formula gives
	\[
	\abs*{\sum_{k\in K}\lambda_k}
	=
	\frac{
		\abs{1-e^{2\pi im/n}}
	}{
		\abs{1-e^{2\pi i/n}}
	}
	=
	\frac{\sin(\pi m/n)}{\sin(\pi/n)}.
	\]
	If \(n\) is even, the roots occur in antipodal pairs, so \(m=n/2\).
	If \(n\) is odd, the right half-plane contains either
	\((n-1)/2\) or \((n+1)/2\) roots. Consequently,
	\[
	\abs{q_{1n}}
	=
	\begin{cases}
		\displaystyle
		\frac{1}{n\sin(\pi/n)},
		&n\text{ even},\\[2mm]
		\displaystyle
		\frac{1}{2n\sin(\pi/(2n))},
		&n\text{ odd}.
	\end{cases}
	\]
	In either case, \(\sin x\leq x\) implies
	\begin{equation}
		\label{eq:fixed_dim_Q_entry}
		\abs{q_{1n}}\geq\frac1\pi.
	\end{equation}
	
	Let $\sign^+(A)=(p_{ij})_{i,j=1}^n$. 
    By \eqref{eq:fixed_dim_projection_similarity},
	\[
	p_{1n}=s^{n-1}q_{1n}.
	\]
	Therefore
	\[
	\norm{\sign^+(A)}
	\geq
	\abs{p_{1n}}
	\geq
	\frac{s^{n-1}}{\pi}.
	\]
	On the other hand, since \(Q\) is an orthogonal projection,
	\[
	\norm{\sign^+(A)}
	\leq
	\norm D\norm Q\norm{D^{-1}}
	=
	s^{n-1}.
	\]
	Thus
	\begin{equation}
		\label{eq:fixed_dim_projection_bounds}
		\frac{s^{n-1}}{\pi}
		\leq
		\norm{\sign^+(A)}
		\leq
		s^{n-1}.
	\end{equation}
	
	It remains to estimate \(\omega_A\). From
	\(A^n=e^{i\theta}s^{-n}I\), one can verify that
	\[
	(A-itI)^{-1}
	=
	\frac{
		A^{n-1}
		+itA^{n-2}
		+\cdots
		+(it)^{n-1}I
	}{
		e^{i\theta}s^{-n}-(it)^n
	}.
	\]
	If \(n\) is odd, then \(e^{i\theta}s^{-n}=s^{-n}\) is real and
	\((it)^n\) is purely imaginary. If \(n\) is even, then
	\(e^{i\theta}s^{-n}=is^{-n}\) is purely imaginary and
	\((it)^n\) is real. Hence, setting \(u=\abs t\),
	\[
	\abs{e^{i\theta}s^{-n}-(it)^n}
	=
	\sqrt{s^{-2n}+u^{2n}}.
	\]
	Since \(\norm A=1\), it follows that
	\begin{equation}
		\label{eq:fixed_dim_resolvent_family}
		\norm{(A-itI)^{-1}}
		\leq
		\frac{ \sum_{k=0}^{n-1}u^k}
		{\sqrt{s^{-2n}+u^{2n}}}.
	\end{equation}
	
	Recall that \(s\geq2\), so \(2^n\leq s^n\). If \(u\leq1/2\),
	then
	\[
	\frac{ \sum_{k=0}^{n-1}u^k}
	{\sqrt{s^{-2n}+u^{2n}}}
	\leq
	2s^n.
	\]
	If \(1/2\leq u\leq1\), then the denominator is at least \(u^n\),
	and
	\[
	\frac{ \sum_{k=0}^{n-1}u^k}{u^n}
	=
	\sum_{j=1}^n u^{-j}
	<
	2^{n+1}
	\leq
	2s^n.
	\]
	Finally, if \(u\geq1\), then
	\[
	\frac{ \sum_{k=0}^{n-1}u^k}{u^n}
	=
	\sum_{j=1}^n u^{-j}
	\leq
	n
	\leq
	2^n
	\leq
	s^n.
	\]
	Thus \eqref{eq:fixed_dim_resolvent_family} gives
	\(
	\sup_{t\in\R}\norm{(A-itI)^{-1}}
	\leq
	2s^n,
	\)
	and hence
	\[
	\omega_A\geq\frac12s^{-n}.
	\]
	Taking \(t=0\) gives
	\[
	\omega_A
	\leq
	\sigma_{\min}(A)
	=
	s^{-n}.
	\]
	Together with \eqref{eq:fixed_dim_projection_bounds}, this shows
	that
	\[
	\frac{1}{\pi2^{(n-1)/n}}
	\omega_A^{-(n-1)/n}
	\leq
	\norm{\sign^+(A)}
	\leq
	\omega_A^{-(n-1)/n}.
	\]

\end{document}